\documentclass[12pt]{article}
\usepackage[utf8]{inputenc}
\usepackage{amsmath}
\usepackage{graphicx}
\usepackage{amssymb}
\usepackage{adjustbox}
\usepackage{multirow}
\usepackage{rotating}
\usepackage{todonotes}
\usepackage{pdflscape}
\usepackage{natbib}
\usepackage{xcolor,color,amsthm}
\usepackage{colortbl}
\usepackage{import}
\presetkeys{todonotes}{size=\setstretch{1}\selectfont\tiny}{}
\date{\today}
\usepackage{authblk}
\usepackage{setspace}
\usepackage{etoolbox}

\author[a]{Anna Neufeld}
\author[b]{Ameer Dharamshi}
\author[c]{Lucy L. Gao}
\author[b,d]{Daniela Witten}
\affil[a]{Public Health Sciences Division, Fred Hutchinson Cancer Center}
\affil[b]{Department of Biostatistics, University of Washington}
\affil[c]{Department of Statistics, University of British Columbia}
\affil[d]{Department of Statistics, University of Washington}

\usepackage[boxruled]{algorithm2e}

\usepackage[margin=1in]{geometry}
\usepackage{float}
\usepackage{hyperref}

\DeclareMathOperator{\E}{\text{E}}
\DeclareMathOperator{\Var}{\text{Var}}
\DeclareMathOperator{\Cov}{\text{Cov}}

\newcommand{\N}{\mathrm{N}}
\newcommand{\eps}{\epsilon}
\newcommand{\xo}{X^{(1)}}

\newcommand{\xtr}{X^{(\mathrm{train})}}
\newcommand{\xte}{X^{(\mathrm{test})}}
\newcommand{\xt}{X^{(2)}}
\newcommand{\xm}{X^{(m)}}

\newcommand{\xM}{X^{(M)}}

\DeclareSymbolFont{yhlargesymbols}{OMX}{yhex}{m}{n}\DeclareMathAccent{\yhwidehat}{\mathord}{yhlargesymbols}{"62}

\newtheorem{proposition}{Proposition}
\newtheorem{definition}{Definition}

\newtheorem{example}{Example}[section]

\newtheorem{corollary}{Corollary}
\newtheorem{theorem}{Theorem}
\newtheorem{remark}{Remark}
\newtheorem{lemma}{Lemma}
\title{Data thinning for convolution-closed distributions}

\DeclareSymbolFont{yhlargesymbols}{OMX}{yhex}{m}{n}\DeclareMathAccent{\yhwidehat}{\mathord}{yhlargesymbols}{"62}

\begin{document}
\LinesNumbered
\SetAlgoSkip{bigskip}
\SetKwInOut{Input}{Input}
\SetKwInOut{Output}{Output}

\title{Data thinning for convolution-closed distributions}

\maketitle

\begin{abstract}%
We propose data thinning, an approach for splitting an observation into two or more independent parts that sum to the 
original observation, and that follow the same distribution as the original observation, up to a (known) scaling of a parameter. This very general proposal is applicable to any  convolution-closed distribution, a  class that includes the Gaussian, Poisson, negative binomial, gamma, and binomial distributions, among others. Data thinning has a number of applications to model selection, evaluation, and inference. For instance, cross-validation via data thinning provides an attractive alternative to the usual approach of cross-validation via sample splitting, especially in settings in which the latter is not applicable. In simulations and in an application to single-cell RNA-sequencing data, we show that data thinning can be used to validate the results of unsupervised learning approaches, such as k-means clustering and principal components analysis, for which traditional sample splitting is unattractive or unavailable. 
\end{abstract}


\section{Introduction}
\label{sec1:intro}

As scientists fit increasingly complex models to their data, there is an ever-growing need for out-of-box methods that can be used to validate these models. In many settings, the most natural option is sample splitting, in which the $n$ observations in a dataset are split into a training set, used to fit a model, and a test set, used to validate it \citep{hastie2009elements}. Sample splitting can also be applied to conduct inference after model selection \citep{rinaldo2019bootstrapping}. Sample splitting is flexible and intuitive, and is a vital tool for any practicing data analyst. 

However, in settings where there is one parameter of interest per observation, or the parameter of interest is a function of the $n$ observations, sample splitting cannot be applied. 
For example, when estimating a low-rank approximation to a matrix, there is one parameter of interest (a latent variable coordinate) for each of the $n$ rows in the matrix. 
Similarly, in fixed-covariate regression under model misspecification, the target parameter depends on the specific $n$ observations included in the dataset \citep{buja2019models}. Finally, there may be settings in which we wish to draw observation-specific inferences about each of our $n$ observations; sample splitting does not allow this.

In this paper, we consider an alternative to sample splitting that splits a single observation $X$ into independent parts that follow the same distribution as $X$. Crucially, the fact that we split a \emph{single} observation avoids the pitfalls of sample splitting in the situations mentioned above: we will split every observation in the data to obtain a training set and a test set that each involve all $n$ observations.

The concept of splitting a single observation as an alternative to sample splitting has been explored in several recent papers \citep{rasines2021splitting, leiner2021data, oliveira2021unbiased, oliveira2022coupled, neufeld2022inference}. We can split $X \sim \N(\mu,\sigma^2)$ with known $\sigma^2$ into two independent Gaussian random variables \citep{rasines2021splitting,leiner2021data,oliveira2021unbiased}, and 
 $X \sim \text{Poisson}(\lambda)$ into two independent Poisson random variables \citep{neufeld2022inference,leiner2021data, oliveira2022coupled}. However, outside of these two distributions, no proposals are available to split a random variable into independent parts that follow the same distribution as the original random variable. \cite{leiner2021data} propose data fission, a general-purpose approach to decompose  $X$ into two parts, $\xo$ and $\xt$, such that (i) $\xo$ and $\xt$ can together be used to reconstruct $X$, and (ii) the joint distribution of $\left(\xo, \xt\right)$ is tractable. However, outside of the special cases of the Gaussian and Poisson distributions, the proposal of \cite{leiner2021data} leads to $\xo$ and $\xt$ that are not independent, and that may not follow the same distribution as $X$. Consequently, unless the data are Gaussian- or Poisson-distributed, data fission does not serve as a direct alternative to sample splitting, as many procedures that would be trivial under sample splitting become very complicated. We elaborate on these points in Section~\ref{appendix:leiner}.

 In this paper, we propose data thinning, a recipe for decomposing a single observation $X$ into two  parts, $\xo$ and $\xt$, such that (i) $X=\xo+\xt$, (ii) $\xo$ and $\xt$ are independent, and (iii) $\xo$ and $\xt$ follow the same distribution as $X$, up to a (known) scaling of a parameter.
 Critically, properties (ii) and (iii) guarantee that  this decomposition is straightforward to use in applied settings. For instance, to evaluate the suitability of a model
 for $X$, we can fit it to $\xo$ (since it follows the same distribution as $X$), and can validate it using $\xt$ (since it also follows the same distribution, and furthermore is independent of $\xo$).  Our recipe can be applied to any distribution that is convolution-closed \citep{joe1996time}: this includes the multivariate Gaussian, Poisson, negative binomial, 
 gamma,  binomial, and multinomial distributions, among others. Thus, our work drastically expands the set of distributions that can be split into independent parts,  and provides a unified lens through which to view seemingly unrelated approaches. 
 Furthermore, data thinning can be used to decompose $X$ into more than two independent random variables. 
 
We illustrate our proposal with the following example, which shows that a gamma random variable can be thinned into $M$ independent gamma random variables. 

\begin{example}[Gamma decomposition into $M$ components, data thinning] \label{ex:gcs-gam-Kfold} Suppose that $X \sim \mathrm{Gamma}(\alpha,\beta)$, where $\beta$ is unknown. 
We take $(X^{(1)},\ldots,X^{(M)}) = XZ$, where  $Z \sim \mathrm{Dirichlet}(\alpha/M,\ldots,\alpha/M)$. Then $X^{(1)},\ldots,X^{(M)}$ are mutually independent, they sum to $X$, and each is marginally drawn from a  $\mathrm{Gamma}(\alpha/M, \beta)$ distribution. 
\end{example}
In other words, data thinning allows us to decompose a $\mathrm{Gamma}(\alpha,\beta)$ random variable, for which $\beta$ is unknown, into $M$ independent gamma random variables, $\xo,\ldots,\xM$. Therefore, fitting a model to $X-\xm$ and validating it using $\xm$ is straightforward. 

The rest of this paper is organized as follows. In Section~\ref{sec:proposal}, we briefly review the class of convolution-closed distributions, and introduce a thinning procedure to split a single random variable into two independent random variables, each of which follows the same distribution as the original random variable (up to a scaling of the parameter(s)). We extend this thinning procedure to split a single random variable into an arbitrary number of independent random variables in Section~\ref{sec:multiplefolds}. We elaborate on the comparison between sample splitting and data thinning in Section~\ref{sec_roleEps}, and in Section~\ref{appendix:leiner} we elaborate on the comparison between data fission \citep{leiner2021data} and data thinning. In Section~\ref{sec:sim}, we focus on validating the results of clustering and low-rank matrix approximations. These are two settings in which the usual cross-validation via sample splitting approach cannot be directly applied \citep[see, e.g.][]{owen2009bi, fu2020estimating}, but data thinning provides a simple alternative. An application to single-cell RNA-sequencing data is in Section~\ref{sec:data}. We close with a discussion in Section~\ref{sec:discussion}. All proofs are in the appendix. 

\section{The data thinning proposal}
\label{sec:proposal}

\subsection{A review of convolution-closed distributions}
\label{subsec:convclosed}

We begin by defining a convolution-closed distribution \citep{joe1996time,jorgensen1998stationary}.

\begin{definition}[Convolution-closed]
	Let $F_\lambda$ denote a distribution indexed by a parameter $\lambda$ in parameter space $\Lambda$. Let ${X}' \sim F_{\lambda_1}$ and ${X}'' \sim F_{\lambda_2}$ with ${X}' \perp\!\!\!\perp {X}''$. If ${X}'+{X}'' \sim F_{\lambda_1+\lambda_2}$ whenever $\lambda_1+\lambda_2 \in \Lambda$, then $F_\lambda$ is \emph{convolution-closed} in the parameter $\lambda$. 
\end{definition} 

Many well-known distributions are convolution-closed. While the $\mathrm{Poisson}(\lambda)$ distribution is convolution-closed in its single parameter $\lambda$ and the $\mathrm{N}(\mu, \sigma^2)$ distribution is convolution-closed in the two-dimensional parameter $(\mu, \sigma^2)$, other distributions, such as the gamma, are convolution-closed in just one parameter with the other parameter(s) held fixed. Table~\ref{tab:convclosed} provides details about some well-known convolution-closed distributions. The following definition provides a useful property of most convolution-closed distributions. 

\begin{definition}[Linear expectation property]
\label{def:linExp}
Let $F_{\lambda}$ denote a distribution indexed by $\lambda \in \Lambda$.  We say that it satisfies the linear expectation property if, for $X \sim F_\lambda$, $E\left[ X \right]$ is a linear function of $\lambda$. 
\end{definition}

\begin{remark}[Most convolution-closed distributions satisfy the linear expectation property]
Let $F_{\lambda}$ be a convolution-closed distribution whose first moment exists for all $\lambda \in \Lambda$. By definition, if $X' \sim F_{\lambda_1}$ and $X'' \sim F_{\lambda_2}$ and $\lambda_1+\lambda_2 \in \Lambda$, then $X' + X'' \sim F_{\lambda_1+\lambda_2}$. By properties of the expected value, $\E[X'+X''] = \E[X'] + \E[X'']$. Thus, the expectation is additive in $\lambda$, and so, outside of contrived counterexamples, $F_\lambda$ satisfies the linear expectation property. The linear expectation property is satisfied for all distributions in Table~\ref{tab:convclosed}. 
\end{remark}

\begin{remark}[Not all convolution-closed distributions satisfy the linear expectation property]
The $\mathrm{Cauchy}(\mu, \gamma)$ distribution is convolution-closed in the two-dimensional parameter $(\mu, \gamma)$, but does not satisfy the linear expectation property because $\E[X]$ does not exist. 
\end{remark}

\begin{table}
\caption{A partial list of convolution-closed distributions. The last two rows contain multivariate distributions. The results in each row are easily verifiable. The generalized Poisson and Tweedie distributions are written in their additive exponential dispersion family parameterization; see \cite{jorgensen1998stationary} for details.}
\centering
\scriptsize
\begin{tabular}{l l}
\centering{Distribution} & Notes \\
\hline
$X \sim \mathrm{Poisson}(\lambda)$, where $\E[X] = \lambda$ and $\Var(X)=\lambda$. & Convolution-closed in $\lambda$. \\
$X \sim \mathrm{N}(\mu, \sigma^2)$, where $\E[X]=\mu$ and $\Var[X]=\sigma^2$. & Convolution-closed in $(\mu, \sigma^2)$. \\
$X \sim \mathrm{NegativeBinomial}(r,p)$, where $\E[X] = r \frac{1-p}{p}$ and $\Var[X] = r \frac{1-p}{p^2}$. & Convolution-closed in $r$ if $p$ is fixed.  \\
$X \sim \mathrm{Gamma}(\alpha, \beta)$, where $\E[X] = \frac{\alpha}{\beta}$ and $\Var(X) = \frac{\alpha}{\beta^2}$. & Convolution-closed in $\alpha$ if $\beta$ is fixed. \\
$X \sim \mathrm{Binomial}(r,p)$, where $\E[X]=rp$ and $\Var(X) = rp(1-p)$. & Convolution-closed in $r$ if $p$ is fixed. \\
$X \sim \mathrm{InverseGaussian}(\mu w, \lambda w^2)$  with $\E[X]=\mu w$ and $\Var(X) = \frac{w^3 \mu^3}{w^2 \lambda} = \frac{w \mu^3}{\lambda}$. & Convolution-closed in $w$ if $\mu$ and $\lambda$ are fixed.  \\
$X \sim \mathrm{GeneralizedPoisson}(\lambda, \theta)$, see \cite{jorgensen1998stationary} for parameterization.&  Convolution-closed in $\lambda$ if $\theta$ is fixed.  \\
$X \sim \mathrm{Tweedie}_p(\lambda,\theta), $ see \cite{jorgensen1998stationary} for parameterization. & Convolution-closed in $\lambda$ if $\theta$ and $p$ are fixed. \\
${X} \sim \mathrm{N}_k \left( {\mu}, {\Sigma}\right)$, with $\E[{X}] = {\mu}$ and $\Var({X}) = {\Sigma}$. & Convolution-closed in $({\mu}, {\Sigma})$. \\
${X} \sim \mathrm{Multinomial}_k \left( r,  {p} \right)$, with $\E[{X}] = r {p}$ and $\Var({X}) = r \left( \mathrm{diag}({p}) - {p}{p}^T \right)$. &  Convolution-closed in $r$ if ${p}$ is fixed.  \\
\end{tabular}
\label{tab:convclosed}	
\end{table}

For a convolution-closed distribution $F_\lambda$, suppose that $X' \sim F_{\lambda_1}$ and $X'' \sim F_{\lambda_2}$ with \\ ${X}' \perp\!\!\!\perp {X}''$. Let $G_{\lambda_1, \lambda_2, x}$ denote the conditional distribution of $X' \mid X'+X'' = x$. The density of the distribution $G_{\lambda_1,\lambda_2,x}$ can be written down for any $F_\lambda$ with a known density function \citep{jorgensen1992exponential}. Furthermore, it turns out that $G_{\lambda_1,\lambda_2,x}$ has a simple closed form for several of the well-known distributions from Table~\ref{tab:convclosed};  see Table~\ref{tab:maintable}. 
 For example, if $F_\lambda$ is the $\mathrm{Poisson}(\lambda)$ distribution, then $G_{\lambda_1, \lambda_2, x}$ is the $\mathrm{Binomial}\left(x, \lambda_1/(\lambda_1+\lambda_2) \right)$ distribution. 
 
\subsection{Data thinning}
\label{subsec:datathin}

Recall from Section~\ref{subsec:convclosed} that  $G_{\lambda_1, \lambda_2, x}$ is the conditional distribution of $X' \mid X'+X'' = x$, where ${X}' \sim F_{\lambda_1}$ and ${X}'' \sim F_{\lambda_2}$ with ${X}' \perp\!\!\!\perp {X}''$.  We now introduce our proposal.

\begin{algorithm}[H]
\caption{Data thinning}
\label{alg:datathin}
\Input{A realization $x$ of $X \sim F_\lambda$, where $F_\lambda$ is convolution-closed in $\lambda$ with parameter space $\Lambda$. A scalar $\epsilon \in (0,1)$ such that $\epsilon \lambda \in \Lambda$ and $(1-\epsilon)\lambda \in \Lambda$. }
Draw $\xo \mid {X}=x \sim G_{\epsilon \lambda, (1-\epsilon)\lambda, x}$. \\
Let $\xt = X - \xo$.	\\ 
\Output{$\left(\xo, \xt\right)$.}
\end{algorithm}

We now introduce our main theorem, which is motivated by a proposal by \cite{joe1996time} to construct autoregressive time series processes with known marginal distributions. 

\begin{theorem}
\label{theorem:datathin}
Suppose that we apply Algorithm~\ref{alg:datathin} to a realization $x$ of $X \sim F_\lambda$. Then, the following results hold:
(i) $\xo \sim F_{\epsilon \lambda}$ and $\xt \sim F_{(1-\epsilon) \lambda}$; (ii) $\xo \perp\!\!\!\perp \xt$; (iii) If $F_\lambda$ satisfies the linear expectation property (Definition~\ref{def:linExp}), then $\E[\xo] = \eps \E[X]$ and $\E[\xt] = (1-\eps) \E[X]$. 
\end{theorem}

Theorem~\ref{theorem:datathin} is proven in Appendix~\ref{appendix:mainproof}. The intuition for parts (i) and (ii) is as follows: if $X \sim F_\lambda$,  then $X$ could have arisen as the sum of two independent random variables $X' \sim F_{\lambda_1}$ and $X'' \sim F_{\lambda_2}$, with $\lambda_1+\lambda_2 = \lambda$. Algorithm~\ref{alg:datathin} works backwards to undo this sum by generating 
$\xo$ and $\xt$ that follow the same distribution as $X'$ and $X''$. Part (iii) follows from Definition~\ref{def:linExp}. As we will see in Section~\ref{subsec:eps}, $\eps \in (0,1)$ is a tuning parameter that governs a tradeoff between how much information is in $\xo$ as opposed to $\xt$.

Theorem~\ref{theorem:datathin} guarantees that the decomposition provided by Algorithm~\ref{alg:datathin} satisfies the goals given in Section~\ref{sec1:intro}: namely $X=\xo+\xt$,  $\xo \perp\!\!\!\perp \xt$, and $\xo$ and $\xt$ follow the same distribution as $X$, up to a (known) scaling of a parameter. Table~\ref{tab:maintable} summarizes the data thinning proposal for several well-known distributions. However, Algorithm~\ref{alg:datathin} and Theorem~\ref{theorem:datathin} apply well beyond the set of distributions considered in Table~\ref{tab:maintable}.

\begin{remark} 
\label{remark:sample}
 In Table~\ref{tab:maintable}, we focus on distributions where the conditional distribution $G_{\lambda_1,\lambda_2, x}$ has a recognizable form. For distributions where this is not the case, standard numerical sampling algorithms can be used to generate $\xo$ and $\xt$, so long as the conditional distribution can be expressed up to a normalizing constant.
 \end{remark}

\begin{remark} 
\label{remark:nuisance}
Some of the decompositions presented in Table~\ref{tab:maintable} require knowledge of an additional parameter that is not of primary interest. For example,  
thinning the $\mathrm{N}(\mu, \sigma^2)$ distribution requires knowledge of $\sigma^2$ \citep[see also:][]{rasines2021splitting, leiner2021data}. 
In Section~\ref{subsec:nuissance}, we explore the implications of performing data thinning in the presence of an unknown nuisance parameter. 
\end{remark}

\begin{remark}\label{remark:binomial}
 Table~\ref{tab:maintable} indicates that thinning the $\mathrm{Binomial}(r,p)$ distribution or the $\mathrm{Multinomial}(r,p)$ distribution requires that $\epsilon r$ take on an integer value. This is because these distributions are not infinitely divisible \citep{joe1996time}. This restriction becomes more limiting in the extension to multiple folds given in Section~\ref{sec:multiplefolds}, and prevents us from thinning the Bernoulli or categorical distributions. 
  \end{remark}

\begin{remark}
  \label{remark:normal_equivalence}
The thinning recipe for the Gaussian can be shown to be equivalent, up to a simple rescaling by $\epsilon$, to a procedure for splitting a Gaussian random variable with known variance that has been used in several recent papers \citep{rasines2021splitting, leiner2021data, oliveira2021unbiased, tian2018selective, tian2020prediction}. We derive this equivalence in Section~\ref{appendix:leiner}. 
  \end{remark}
  
 We now give an example of an application where data thinning is useful in practice. 
  
\begin{example}[Model validation using data thinning]
\label{ap:Eval}
Suppose we observe ${X}_{ij}$ for $i=1,\ldots,n$ and $j=1,\ldots,d$, where either each ${X}_{ij}$ is drawn independently from a univariate convolution-closed distribution that satisfies the linear expectation property from Definition~\ref{def:linExp}, or else each row $(X_{i1},\ldots, X_{id})^T$ is drawn independently from a multivariate convolution-closed distribution that satisfies the linear expectation property. 

We wish to evaluate $\hat{\mu}({X})$ as an estimator for $\E[{X}]$. Computing a loss function such as mean-squared error between $\hat{\mu}({X})$ and ${X}$ is unsatisfactory, since the loss function will take on a small value 
if we overfit the mean.
Instead, we apply Algorithm~\ref{alg:datathin} with $\epsilon \in (0,1)$ to 
either each element or each row in $X$, such that each element 
${X}_{ij}$ is thinned into ${X}_{ij}^{(1)}$ and ${X}_{ij}^{(2)}$. We compute $\hat{\mu}({X}^{(1)})$, which is an estimator of $\E[{X}^{(1)}]=\epsilon \E[{X}]$ (Theorem~\ref{theorem:datathin}, part (iii)). 
We then compute a loss function between $\hat{\mu}({X}^{(1)})$ and ${X}^{(2)}$. Since $\xo \perp\!\!\!\perp \xt$ (Theorem~\ref{theorem:datathin}, part (ii)), the loss function will not take on small values due to overfitting. 
\end{example}

In Example~\ref{ap:Eval}, if $\epsilon = 0.5$, then $\E[{X}^{(1)}] = \E[{X}^{(2)}] = 0.5 \E[{X}]$. Thus, $\hat{\mu}({X}^{(1)})$ is an estimator of $\E[{X}^{(2)}]$, and so devising a suitable loss function is straightforward. If $\epsilon \neq 0.5$, then $\hat{\mu}({X}^{(1)})$ is a plug-in estimator of $\frac{\epsilon}{1-\epsilon} \E[{X}^{(2)}]$ (Theorem~\ref{theorem:datathin}). The following example shows how a loss function that takes into account this factor of $\epsilon$ can be constructed in practice. 

\begin{example}[Example~\ref{ap:Eval} with mean squared error  loss]
\label{ex:lossMSE}
Suppose we wish to use mean squared error to define a loss function between $\hat{\mu}(\xo)$ and $\xt$ in Example~\ref{ap:Eval}.  Since $\E[{X}^{(2)}] = \frac{1-\epsilon}{\eps} \E[{X}^{(1)}]$, we compute the loss as
$$
\frac{1}{n d} \left\| {X}^{(2)} - \frac{1-\eps}{\eps} \hat{\mu}({X}^{(1)})\right\|_F^2,
$$
where the factor of $ (1-\epsilon)/\eps$  turns an estimate of $\E[{X}^{(1)}]$ into an estimate of $\E[{X}^{(2)}]$.
\end{example}

We discuss the choice of $\epsilon$ in Section~\ref{subsec:eps}. 

\begin{table}
\caption{Details of data thinning for several well-known distributions, using the parameterizations given in Table~\ref{tab:convclosed}. While the exponential distribution itself is not convolution-closed in its single parameter, recognizing it as a special case of the gamma distribution with known $\alpha=1$ yields a decomposition. In all cases, the distribution of $X^{(2)}$ matches that of $X^{(1)}$, with $\epsilon$ replaced by $(1-\epsilon)$, with $\xo \perp\!\!\!\perp \xt$.
}
\scriptsize
\begin{tabular}{l  l l  l}
Distribution of ${X}$ & Generate $\xo \mid {X}=x $ as: &  Dist. of $X^{(1)}$ & Notes \\
\hline
$\text{Poisson}(\lambda)$ & Draw $\xo \mid {X}=x \sim \text{Binomial}(x, \epsilon).$ & $\text{Poisson}(\epsilon \lambda)$ &  \\
$\text{N}(\mu, \sigma^2)$ & Draw $\xo \mid {X}=x \sim \text{N}(\epsilon x, \epsilon (1-\epsilon) \sigma^2).$ & $\text{N}(\epsilon \mu, \epsilon \sigma^2)$ &  $\sigma^2$ must be known. \\
$\mathrm{NegativeBinomial}(r, p)$ &   Draw $\xo \mid {X}=x \sim \text{BetaBinomial}(x, \epsilon r, (1-\epsilon) r). $ & $\text{NegativeBinomial}(\epsilon r, p)$ & $r$ must be known.  \\
$\text{Gamma}(\alpha, \beta)$ & Draw $Z \sim \mathrm{Beta}\left(\epsilon \alpha, (1-\eps) \alpha \right)$, and let $\xo = x \cdot Z$.  & $\text{Gamma}(\epsilon \alpha, \beta)$  & $\alpha$ must be known. \\
$\text{Exponential}(\lambda)$ & Draw $Z \sim \mathrm{Beta}(\epsilon, (1-\epsilon))$, and let $\xo = x \cdot Z$. & $\text{Gamma}(\epsilon, \lambda)$  &  \\ 
$\text{Binomial}(r,p)$ & Draw $\xo \mid {X}=x \sim \text{Hypergeometric}(\epsilon r, (1-\epsilon)r, x).$ & $\text{Binomial}(\epsilon r,p)$ & $r$ must be known \\
&&& $\epsilon r$ must be integer. \\
$\text{N}_k({\mu},{\Sigma})$ & Draw ${X}^{(1)} \mid {{X}}={x} \sim \text{N}(\epsilon {x}, \eps (1-\eps) {\Sigma}).$ & $\text{N}_k(\eps {\mu} , \eps {\Sigma})$ & $\Sigma$ must be known.  \\
$\text{Multinomial}_k(r, {p})$ & Draw ${X}^{(1)} \mid {{X}}={x} \sim$ & $\text{Multinomial}_k(\eps r, {p})$ & $r$ must be known.  \\
& $\text{MultivariateHypergeometric}({x}_1, {x}_2, \ldots, {x}_k, \eps r)$. & & $\epsilon r$ must be integer. \\
\\
\end{tabular}	
\label{tab:maintable}
\end{table}

\subsection{Effect of unknown nuisance parameters}
\label{subsec:nuissance}

For several of the distributions in Table~\ref{tab:maintable}, data thinning requires knowledge of a nuisance parameter. For example, thinning a $\mathrm{N}(\mu, \sigma^2)$ distribution requires knowledge of  $\sigma^2$. 

We now consider what happens when we perform data thinning on Gaussian data using an incorrect value of 
the variance. We refer to this incorrect value as $\tilde{\sigma}^2$.

\begin{proposition}
Suppose that we observe $x$ from $X \sim \mathrm{N}(\mu, \sigma^2)$.  We draw $\xo \mid X=x \sim \mathrm{N}\left(\epsilon x, \epsilon (1-\epsilon) \tilde{\sigma}^2 \right)$, for some $\tilde{\sigma}^2$ that is not a function of $x$, and let $\xt = X-\xo$. Then:
(i) $\xo \sim \mathrm{N}\left(\epsilon \mu,  \epsilon^2 \sigma^2 +  \epsilon (1-\epsilon) \tilde{\sigma}^2 \right)$, 
(ii) $\xt \sim \mathrm{N}\left((1-\epsilon) \mu, (1-\epsilon)^2 \sigma^2 +  \epsilon (1-\epsilon) \tilde{\sigma}^2\right)$,  and 
(iii) $\mathrm{cov}\left(\xo, \xt\right) = \epsilon (1-\epsilon) \left( \sigma^2 - \tilde{\sigma}^2 \right)$. 
\label{prop:normalnuissance}
\end{proposition}
Part (iii) of Proposition~\ref{prop:normalnuissance} indicates that if we apply data thinning with too little noise  ($\tilde{\sigma}^2 < \sigma^2$), then $\xo$ and $\xt$ are positively correlated. On the other hand, if we apply data thinning with too much noise ($\tilde{\sigma}^2 > \sigma^2$), then $\xo$ and $\xt$ are negatively correlated. Similar results hold for the negative binomial distribution and the gamma distribution. 

\begin{proposition}
Suppose that we observe  $x$ from $X \sim \mathrm{NegativeBinomial}(r, p)$. We draw $\xo \mid X=x \sim \mathrm{BetaBin}\left(x, \epsilon \tilde{r}, (1-\epsilon) \tilde{r}\right)$ for some  $\tilde{r}$ that is not a function of $x$, and let $\xt = X-\xo$. Then
$
\mathrm{cov}\left(\xo, \xt \right) = \epsilon(1-\epsilon) r \left(\frac{1-p}{p}\right)^2 \left( 1 - \frac{r+1}{\tilde{r}+1}\right). 
$
\label{prop:nbnuissance}
\end{proposition}

\begin{proposition}
Suppose that we observe  $x$ from $X \sim \mathrm{Gamma}(\alpha, \beta)$. We let $\xo = x \times Z$, where $Z \sim \mathrm{Beta}\left(\epsilon \tilde{\alpha}, (1-\epsilon) \tilde{\alpha}\right)$ for some $\tilde{\alpha}$ that is not a function of $x$. We let $\xt = X - \xo$. Then
$
\mathrm{cov}\left(\xo, \xt \right) = \eps (1-\eps) \frac{\alpha}{\beta^2} \left( 1 - \frac{\alpha+1}{\tilde{\alpha}+1}\right).
$
\label{prop:gamma_nuissance}
\end{proposition}

Propositions~\ref{prop:normalnuissance}--\ref{prop:gamma_nuissance} are proven in Appendix~\ref{appendix:nuissanceproof}.  Figure~\ref{fig_corVis} verifies these results empirically. The results in this section assume that $\tilde\sigma^2$, $\tilde{r}$, and $\tilde{\alpha}$ are not a function of $x$. In practice, one might estimate the unknown parameters $\sigma^2$, $r$, and $\alpha$ using additional data. 

\begin{figure}
\centering
\includegraphics[width=0.8\textwidth]{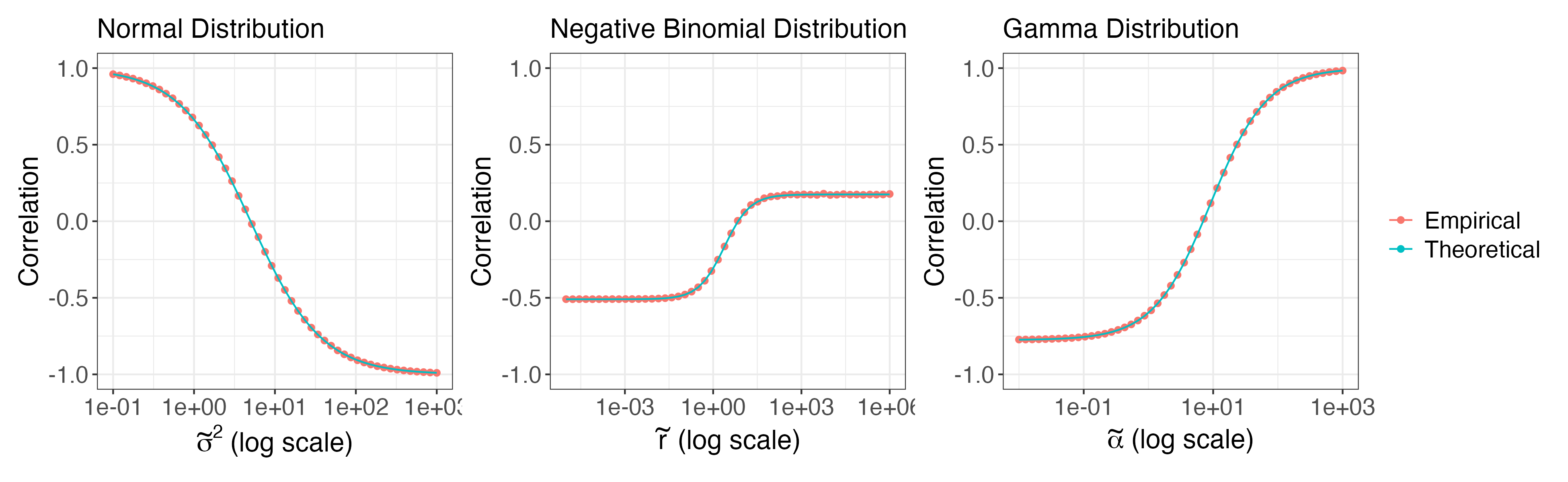}	
\caption{\emph{Left:} We generate 100,000 realizations of $X \sim \mathrm{N}(7,5)$. For 50 values of $\tilde{\sigma}^2$, we thin $X$ into $\xo$ and $\xt$ using $\tilde{\sigma}^2$ instead of $\sigma^2=5$. 
\emph{Center:} We generate 100,000 realizations of $X \sim \mathrm{NB}(7,0.7)$. For 50 values of $\tilde{r}$,  we thin $X$ into $\xo$ and $\xt$ using $\tilde{r}$ instead of $r=7$. 
\emph{Right:} We generate 100,000 realizations of $X \sim \mathrm{Gamma}(7,5)$. For 50 values of $\tilde{\alpha}$,  we thin $X$ into $\xo$ and $\xt$ using $\tilde{\alpha}$ instead of $\alpha=7$. 
\emph{All:} In each panel, for each value of the nuisance parameter, we display the empirical correlation between $\xo$ and $\xt$ (red dots), along with the theoretical correlation suggested by Propositions~\ref{prop:normalnuissance}--\ref{prop:gamma_nuissance} (blue lines). In all cases, we use $\epsilon = 0.44$ for thinning.  }
\label{fig_corVis}
\end{figure}

\section{Multifold data thinning}
\label{sec:multiplefolds}

Data thinning involves decomposing $X$ into $\xo$ and $\xt$, which each have the same distribution as $X$ (up to a known parameter scaling). It can be applied recursively to create $M$ independent data folds,  $\xo, \ldots, \xM$, that sum to $X$, as in the following example. 

\begin{example}[Recursive thinning of the normal distribution]
\label{ex:normal_multifold_hard}
Let $x$ denote a realization of $X \sim \N(\mu, \sigma^2)$. Given $\epsilon_1, \epsilon_2, \epsilon_3 \in (0,1)$ with $\epsilon_1+\epsilon_2+\epsilon_3= 1$, we first 	draw $\xo \mid X \sim \mathrm{N}\left(\epsilon_1 X, \epsilon_1 (1-\epsilon_1) \sigma^2\right)$. Let $X^{(2,3)} = X-\xo$. By Theorem~\ref{theorem:datathin}, $\left( \xo, X^{(2,3)} \right) \sim N\left(\eps_1 \mu, \eps_1 \sigma^2\right) \times N\left((1-\eps_1) \mu, (1-\eps_1) \sigma^2\right)$. 

We next draw $\xt \mid X^{(2,3)} \sim \mathrm{N}\left(\frac{\epsilon_2}{1-\eps_1} X^{(2,3)}, \frac{\epsilon_2}{1-\epsilon_1} (1-\frac{\epsilon_2}{1-\epsilon_1}) (1-\epsilon_1) \sigma^2\right)$, and let $X^{(3)} = X-\xo-\xt$. By Theorem~\ref{theorem:datathin}, $\left( \xt, X^{(3)}\right) \sim \N(\eps_2 \mu, \eps_2 \sigma^2) \times N(\eps_3 \mu, \eps_3 \sigma^2)$. 

Furthermore, since $\left( X^{(2)}, X^{(3)}\right)$ is a function of $X^{(2,3)}$, the pair $\left(X^{(2)}, X^{(3)}\right)$ remains independent of $\xo$. Thus, $\left( X^{(1)}, \xt, X^{(3)}\right) \sim N\left(\eps_1 \mu, \eps_1 \sigma^2\right) \times N(\eps_2 \mu, \eps_2 \sigma^2) \times N(\eps_3 \mu, \eps_3 \sigma^2)$. 
\end{example}

While Example~\ref{ex:normal_multifold_hard} can be extended to create $M>3$ folds, this recursive approach can be cumbersome. In Example~\ref{ex:gcs-gam-Kfold} of Section~\ref{sec1:intro}, we saw that, for the gamma distribution, there is a simple way to create multiple folds without recursion. 
  We will now provide a general form of this result. Let $G_{\lambda_1, \lambda_2, \ldots, \lambda_M, x}$ denote the joint distribution of $\left(X_1 , \ldots, X_M \right) \mid {X_1 + X_2 + \ldots + X_M} = x$, where 
$X_m \overset{\mathrm{ind.}}{\sim} F_{\lambda_{m}}$, for $m=1,\ldots,M$, and where $F_\lambda$ is a convolution-closed distribution. The following algorithm and theorem  mimic Algorithm~\ref{alg:datathin} and Theorem~\ref{theorem:datathin}. 

\begin{algorithm}[H]
\caption{Multifold data thinning.}	
\label{alg:datathin_manyfolds}
\Input{A realization $x$ of  $X \sim F_\lambda$, where $F_\lambda$ is a  convolution-closed distribution with parameter space $\Lambda$. Scalars $\epsilon_1,\ldots, \epsilon_M \in (0,1)$ such that $\sum_{m=1}^M \epsilon_m = 1$ and $\epsilon_m \lambda \in \Lambda$ for $m=1,\ldots M$. }
Draw $\left(X^{(1)},\ldots,X^{(M)}\right) \sim G_{\epsilon_1 \lambda, \epsilon_2 \lambda, \ldots, \epsilon_M \lambda, x}$. \\
\Output{$\left(X^{(1)},\ldots,X^{(M)}\right)$}
\end{algorithm}

\begin{theorem}
\label{theorem:datathin_manyfolds}
Suppose we apply Algorithm~\ref{alg:datathin_manyfolds} to a realization $x$ of $X \sim F_\lambda$, for a convolution-closed distribution $F_\lambda$. Then, the following results hold: (i) $X^{(m)}  \sim F_{\epsilon_m \lambda}$ for $m=1,\ldots,M$; (ii) $\xo,\ldots, \xM$ are mutually independent; 
(iii) $\xo+\xt+\cdots+X^{(M)}=X$; and (iv) if $F_{\lambda}$ satisfies the linear expectation property (Definition~\ref{def:linExp}), then $\E[\xm] = \eps_m \E[X]$ for $m = 1,\ldots,M$.
\end{theorem}

The proof of Theorem~\ref{theorem:datathin_manyfolds} is included in Appendix~\ref{appendix:mainproof}, and is a straightforward extension of that of Theorem~\ref{theorem:datathin}. The intuition for parts (i)-(iii) is as follows: we know that $X \sim F_\lambda$ could have arisen as the sum of $M$ mutually independent random variables $X_1,\ldots,X_M$ such that $X_m \sim F_{\epsilon_m \lambda}$. If we draw $\left(X^{(1)},\ldots,X^{(M)}\right)  | X=x \sim G_{\epsilon_1 \lambda, \epsilon_2 \lambda, \ldots, \epsilon_M \lambda, x}$, then the joint distribution of  $\left(X^{(1)},\ldots,X^{(M)}\right)$ equals the joint  distribution of $\left(X_1,\ldots,X_M\right)$, i.e. it is the joint distribution of $M$ independent random variables with distributions $F_{\eps_1 \lambda},\ldots,F_{\eps_M\lambda}$. Part (iv) follows directly from Definition~\ref{def:linExp}. We now revisit the case of the Gaussian distribution from Example~\ref{ex:normal_multifold_hard}.

\begin{example}[Multifold thinning of the normal distribution]
\label{ex:normal_multifold_easy}
Let $X \sim \N(\mu, \sigma^2)$ and let $\epsilon_1, \epsilon_2, \epsilon_3 > 0$ with $\sum_{i=1}^3 \epsilon_i = 1$.  To generate $M=3$ independent folds of the data, we draw
$$
\small
\left[ \begin{matrix}
\xo \\
\xt \\
X^{(3)}	
\end{matrix} \right] \mid X=x \sim N \left( \left[ \begin{matrix}
\eps_1 x \\
\eps_2 x  \\
\eps_3 x
\end{matrix} \right], 
\left[ \begin{matrix}
\eps_1 (1-\eps_1) \sigma^2 & -\eps_1 \eps_2 \sigma^2 & -\eps_1 \eps_3 \sigma^2 \\ 
-\eps_1 \eps_2 \sigma^2 & \eps_2 (1-\eps_2) \sigma^2 & -\eps_2 \eps_3 \sigma^2 \\
-\eps_1 \eps_3 \sigma^2 & -\eps_2 \eps_3 \sigma^2 & \eps_3 (1 - \eps_3) \sigma^2  
\end{matrix} \right]
\right).
$$
One can verify that this multivariate normal corresponds to  $G_{\eps_1 \lambda, \eps_2 \lambda, \eps_3 \lambda, x}$.
By Theorem~\ref{theorem:datathin_manyfolds}, $\xo, \xt,$ and $X^{(3)}$ are independent and  $X^{(m)} \sim \N(\eps_m \mu, \eps_m \sigma^2)$ for $m=1,2,3$. This distribution $G_{\eps_1 \lambda, \eps_2 \lambda, \eps_3 \lambda, x}$ is a degenerate multivariate normal distribution, which enforces the constraint that the realized values of $\xo, \xt$, and $X^{(3)}$ sum to $x$. 
\end{example}

Table~\ref{tab:folds_table} reveals that  $G_{\epsilon_1 \lambda, \epsilon_2 \lambda, \ldots, \epsilon_M \lambda, x}$ in Algorithm~\ref{alg:datathin_manyfolds} has a very simple form for every univariate distribution in Table~\ref{tab:maintable}. We omit the multivariate distributions to avoid cumbersome notation. Once again, in cases where the conditional distribution is not a recognizable distribution, if its density is known up to a normalizing constant we can generate $\xo,\ldots,\xM$ using sampling techniques. 

\begin{table}
\caption{Details of how to perform multifold data thinning (Algorithm~\ref{alg:datathin_manyfolds}) for several common univariate distributions, where where ${\epsilon} = (\eps_1,\ldots,\eps_M)^T$. In the decomposition of the binomial distribution, $\epsilon_m r$ must be an integer. Each row can be verified using properties of these distributions.} 
\scriptsize
\centering
\begin{tabular}{l l l}
Distribution of ${X}$ & Generate $\left(\xo,\ldots,\xM\right) \mid {X}=X$ as: &  Dist. of $X^{(m)}$ \\
\hline
$\text{Poisson}(\lambda)$ & $\left(\xo,\ldots,\xM\right) \mid {X}=x \sim \text{Multinomial}(x, \epsilon_1,\ldots,\epsilon_M) $. &
$\text{Poisson}(\epsilon_m \lambda)$ \\
$\text{N}(\mu, \sigma^2)$ &  $\left(\xo,\ldots,\xM\right)  \mid {X}=x \sim \text{N} \left( \mu \boldsymbol{\eps}, \sigma^2 \mathrm{diag}({\epsilon}) - \sigma^2 {\epsilon} {\epsilon}^T \right)$, & $\text{N}(\epsilon_m \mu, \epsilon_m \sigma^2)$,    \\
$\mathrm{NegativeBinomial}(r, p)$ &   $\left(\xo,\ldots,\xM\right) \mid {X}=x \sim \text{DirichletMultinomial}(X, \epsilon_1 r, \ldots, \epsilon_M  r) $. & $\text{NegativeBinomial}(\epsilon_m r, p)$  \\
$\text{Gamma}(\alpha, \beta)$ & Draw $Z \sim \text{Dirichlet}\left(\epsilon_1 \alpha, \ldots, \epsilon_M \alpha \right)$, and and let $\left(\xo,\ldots,\xM\right) = x \cdot Z$ & $\text{Gamma}(\epsilon_m \alpha, \beta)$  \\
$\text{Exponential}(\lambda)$ & Draw $Z \sim  \mathrm{Dirichlet}(\eps_1, \ldots, \eps_M)$, and let $\left(\xo,\ldots,\xM\right) = x \cdot Z$. & $\text{Gamma}(\epsilon_m, \lambda)$    \\
$\text{Binomial}(r,p)$ & $\left(\xo,\ldots,\xM\right) \mid {X}=x \sim \text{MultivariateHypergeometric}(\epsilon_1 r, \ldots, \epsilon_M r, x)$. & $\text{Binomial}(\epsilon_m r,p)$  \\
\end{tabular}	
\label{tab:folds_table}
\end{table}

The role of the parameters $\eps_1,\ldots,\eps_M$ in Algorithm~\ref{alg:datathin_manyfolds} is discussed in Section~\ref{subsec:eps}. We now consider the following extension of  Example~\ref{ap:Eval}. 

\begin{example}[Cross validation using multifold thinning]
In the setting of Example~\ref{ap:Eval}, we apply Algorithm~\ref{alg:datathin_manyfolds} with parameters $\epsilon_1,\ldots,\epsilon_M$ to either each element or each row of $X$ such that each element ${X}_{ij}$ is thinned into ${X}_{ij}^{(1)}, \ldots, {X}_{ij}^{(M)}$.

Then, for $m = 1,\ldots,M$, we first define $X^{(-m)} := X - X^{(m)}$. We obtain $\hat{\mu}\left( {X}^{(-m)} \right)$, which is an estimator of $\E[{X}^{(-m)}] = (1-\eps_m)\E[{X}]$. We then compute a loss function between $\hat{\mu}\left( {X}^{(-m)} \right)$ and ${X}^{(m)}$. For example, as in Example~\ref{ex:lossMSE}, we can compute the mean squared error between $\frac{\eps_m}{1-\eps_m}\hat{\mu}\left({X}^{(-m)}\right)$ and $\xm$. We evaluate the estimator $\hat{\mu}(\cdot)$ by averaging the loss across folds. 
\label{ap:crossVal}
\end{example}

The advantage of multifold thinning (Example~\ref{ap:crossVal}) over single fold thinning with $\epsilon=1/M$ (Example~\ref{ap:Eval}) is reduction of the variance of the loss function via averaging.  We will demonstrate the practical advantages of multifold thinning in Section~\ref{sec:sim}.

\section{Comparing data thinning and sample splitting}
\label{sec_roleEps}

In comparing data thinning and sample splitting in a particular setting, there are two considerations. First, we must figure out if each method is applicable. Then, in settings where both are applicable, we must figure out if one method is preferable. 

Data thinning requires an assumption that each 
entry in our dataset is drawn from a specific (but possibly different) convolution-closed distribution. Sample splitting requires no such parametric assumption. On the other hand, we cannot apply sample splitting if our task requires estimating a parameter for every individual observation or drawing conclusions about specific observations. For example, suppose that we wish to evaluate the performance of a clustering algorithm. After sample splitting, clustering the observations in the training set does not yield cluster assignments for the observations in the test set, and thus there is nothing to evaluate on the test set \citep[see  e.g.][]{gao2022selective, fu2020estimating}. Similarly, it is not clear how to use sample splitting to validate a low-rank matrix approximation, for which a latent coordinate must be estimated for each of the $n$ observations \citep[see, e.g.][]{owen2009bi}. We focus on these examples, where sample splitting is not an option, in Sections~\ref{sec:sim} and \ref{sec:data}. 

In Section~\ref{subsec:eps}, we show that the parameters $\epsilon_1,\ldots,\epsilon_M$ in multifold data thinning (Algorithm~\ref{alg:datathin_manyfolds}) control the allocation of information across folds. In Section~\ref{subsec:theory_ss}, we use this information allocation result to theoretically argue that sample splitting and data thinning achieve similar performance, but that data thinning may be preferable in settings where the observations are not identically distributed. In particular, we see that sample splitting is not an attractive choice for ``fixed-covariate" regression in the presence of high-leverage points. In Section~\ref{subsec:empirical_ss}, we empirically compare data thinning and sample splitting in a setting where both are applicable. 

\subsection{Role of the parameters $\epsilon_1,\ldots, \epsilon_M$ in multifold data thinning}
\label{subsec:eps}

In Algorithm~\ref{alg:datathin_manyfolds}, the parameters $\epsilon_1,\ldots,\epsilon_M$ determine how the information in a random variable $X$ about an unknown parameter is allocated across folds of data.  

\begin{theorem}
\label{theorem_epsilon}
Suppose that we thin a random variable $X$ using Algorithm~\ref{alg:datathin_manyfolds} with parameters
$\epsilon_1, \ldots, \epsilon_M$ to obtain $\xo, \ldots, \xM$. Let $I_X(\theta)$ denote the Fisher information contained in $X$ about an unknown parameter $\theta$, i.e. a parameter in the distribution of $X$ that does not appear in $G_{\epsilon_1 \lambda, \ldots, \epsilon_M \lambda, x}$. Assume that this Fisher information exists. Then  $I_{X^{(m)}}(\theta) = \epsilon_m I_{X}(\theta)$ for $m=1,\ldots,M$. 
\end{theorem}

\begin{remark}
\label{remark_parameter}
In the context of Theorem~\ref{theorem_epsilon}, the parameter $\theta$ may or may not be a function of the convolution-closed parameter $\lambda$, but it must be a parameter that is unknown during the thinning process. Below, we list a few examples where Theorem~\ref{theorem_epsilon} applies and where it does not. 
\begin{itemize}
	\item \textbf{Poisson distribution:} Let $X \sim \mathrm{Poisson}(\lambda)$, and suppose we thin $X$ to obtain $\xm \sim \mathrm{Poisson}(\eps_m \lambda)$. As $\lambda$ is unknown during the thinning process, Theorem~\ref{theorem_epsilon} says that $I_{\xm}(\lambda) = \epsilon_m I_X(\lambda)$. We can easily verify this by direct calculation: $I_{X}(\lambda) =\frac{1}{\lambda}$ and $I_{\xm}(\lambda) = \frac{\eps_m}{\lambda}$. 
	\item \textbf{Binomial distribution:} Let $X \sim \mathrm{Binomial}(r,p)$, and suppose that we thin $X$ to obtain $\xm \sim \mathrm{Binomial}(\eps r,p)$. Then $I_X(p) = r/(p(1-p))$, and direct calculation verifies that $I_{\xm}(p) = \eps_m I_X(p)$. On the other hand, Theorem~\ref{theorem_epsilon} makes no claims about the parameter $r$, since $r$ must be known during thinning. 
	\item \textbf{Gaussian distribution:} Let $X \sim \mathrm{N}(\mu,\sigma^2)$, and suppose that we thin $X$ to obtain $\xm \sim \mathrm{N}(\eps_m \mu, \eps_m \sigma^2)$. Then direct computation verifies that $I_{\xm}(\mu) = \eps_m I_X(\mu)$. 
	However, Theorem~\ref{theorem_epsilon} makes no claims about the parameter $\sigma^2$, since it must be known during thinning. 
\end{itemize}
\end{remark}

Theorem~\ref{theorem_epsilon} implies that when choosing $\epsilon_1 ,\ldots, \epsilon_M$ for Algorithm~\ref{alg:datathin_manyfolds} or when choosing $\epsilon$ for Algorithm~\ref{alg:datathin}, one should consider how much information to devote to the training task (i.e. model fitting) as opposed to the testing task (i.e. model evaluation). In Example~\ref{ap:Eval}, as $\epsilon$ increases, the quality of the estimator  $\hat{\mu}(\xo)$ increases, but the information available for computing the loss between this estimator and $\xt$ decreases. 

Recall that applying Algorithm~\ref{alg:datathin} to $X$ with parameter $\epsilon=\frac{1}{M}$ yields $\xo$ and $\xt$ with the same distributions as $\xm$ and $X - \xm$ when we apply Algorithm~\ref{alg:datathin_manyfolds} to $X$ with $\epsilon_1=\ldots=\epsilon_M = \frac{1}{M}$. Thus, the information allocation between $\xm$ and $X - \xm$ after multifold thinning is the same as the information allocation between $\xo$ and $\xt$ after two-fold thinning with $\epsilon=\frac{1}{M}$. The advantage of using multiple folds for model validation (i.e. Example~\ref{ap:crossVal} rather than Example~\ref{ap:Eval}) comes from the reduction in variance that results from averaging the loss function across folds.
 
\subsection{Theoretical comparison to sample splitting}
\label{subsec:theory_ss}

In this section, we focus on the case of two-fold thinning (Algorithm~\ref{alg:datathin}) for simplicity. Understanding the way in which the parameter $\epsilon$ in Algorithm~\ref{alg:datathin} partitions information about parameters of interest helps us draw direct connections between data thinning (with parameter $\epsilon$) and sample splitting (with $\epsilon$ denoting the proportion of observations assigned to the training set). 

\begin{corollary}
\label{theorem_iid}
Suppose that we observe $n$ independent and identically distributed (iid) random variables $\bold{X} = \left(X_1,\ldots,X_n \right)$, where $X_i \sim F_{\lambda}$. Assume that $\epsilon n$ is an integer and that $\epsilon \lambda \in \Lambda$. Consider the following two methods for splitting $\bold{X}$ into independent training and test sets. 
\begin{itemize}
	\item \textbf{Sample splitting: } Assign a specific set of $\epsilon n$ observations to the training set, denoted $\bold{X}^{\mathrm{train}}_{\mathrm{ss}}$, and the remaining $(1-\epsilon) n$ observations to the test set, denoted $\bold{X}^{\mathrm{test}}_{\mathrm{ss}}$. 
	\item \textbf{Data thinning: } For $i=1,\ldots,n$, thin $X_i$ into $X_i^{(1)}$ and $X_i^{(2)}$ by applying Algorithm~\ref{alg:datathin} with parameter $\epsilon$. Let $\bold{X}^{\mathrm{train}}_{\mathrm{dt}}=\left(X_1^{(1)}, \ldots, X_n^{(1)}\right)$ be the training set and let $\bold{X}^{\mathrm{test}}_{\mathrm{dt}} = \left(X_1^{(2)}, \ldots, X_n^{(2)}\right) $ be the test set. 
\end{itemize}
Let $\theta$ be an unknown parameter of interest, as in Theorem~\ref{theorem_epsilon}. Then, $I_{\bold{X}^{\mathrm{train}}_{\mathrm{ss}}}(\theta) = I_{\bold{X}^{\mathrm{train}}_{\mathrm{dt}}}(\theta) = \epsilon I_{\bold{X}}(\theta)$, 
and $I_{\bold{X}^{\mathrm{test}}_{\mathrm{ss}}}(\theta) = I_{\bold{X}^{\mathrm{test}}_{\mathrm{dt}}}(\theta) = (1-\epsilon) I_{\bold{X}}(\theta)$.
\end{corollary}

In Corollary~\ref{theorem_iid}, $I_{\bold{X}^{\mathrm{train}}_{\mathrm{ss}}}(\theta)$ takes on the same value for any specific allocation of datapoints to the training set. This means that if we instead randomly allocate data points to the training set, as is typical in practice, the information split remains identical to data thinning. Consequently, we expect sample splitting and data thinning to have similar performance regarding inference on unknown parameters in settings where both are options and where the datapoints are independent and identically distributed. A similar point was made by \cite{leiner2021data}, who view their data fission technique as a ``continuous analog" of sample splitting, since the requirement for sample splitting that $\epsilon n$ be an integer limits the choice of $\epsilon$, especially when $n$ is small. 

We next consider a setting where the observations are independent but not identically distributed, and thus the difference between data thinning and sample splitting is more pronounced. 

\begin{example}[Fixed-covariate regression]
\label{theorem_noniid}
Suppose that we observe a fixed set of covariates $Z_1,\ldots, Z_n$. Suppose that $X_i \overset{\mathrm{ind.}}{\sim} N\left( \beta Z_i, \sigma^2 \right)$ for $i=1,\ldots,n$, and let $\bold{X} = \left(X_1,\ldots,X_n\right)$. In this setting, $I_{X_i}(\beta) = \frac{Z_i^2}{\sigma^2}$, meaning that observations with larger values of $Z_i^2$ (high-leverage points) contain more information about $\beta$, the unknown parameter of interest. Let $I_{\bold{X}^{\mathrm{train}}_{\mathrm{ss}}}(\beta)$ and $I_{\bold{X}^{\mathrm{train}}_{\mathrm{dt}}}(\beta)$ be defined as in Corollary~\ref{theorem_iid}, where the unknown parameter of interest is the slope $\beta$. Then
$$
I_{\bold{X}^{\mathrm{train}}_{\mathrm{dt}}}(\beta) = \sum_{i=1}^n I_{X_i^{(1)}}(\beta)= \sum_{i=1}^n  \epsilon \frac{Z_i^2}{\sigma^2} = \epsilon I_{\bold{X}}(\beta).
$$
However, 
$$
I_{\bold{X}_{\mathrm{ss}}^{\mathrm{train}}}(\beta) = \sum_{i \in \mathrm{train}} I_{X_i}(\beta) = \sum_{i \in \mathrm{train}}  \frac{Z_i^2}{\sigma^2} \neq \epsilon I_{\bold{X}}(\beta),
$$
where $\mathrm{train}$ denotes the specific indices of the observations assigned to the training set. Thus, while data thinning always allocates a fraction $\epsilon$ of the Fisher information to the training set, the information allocation of sample splitting depends on the specific assignment of observations to the training set. 
\end{example}

Example~\ref{theorem_noniid} shows that $I_{\bold{X}_{\mathrm{ss}}^{\mathrm{train}}(\beta)} \neq  I_{\bold{X}^{\mathrm{train}}_{\mathrm{dt}}(\beta)}$ for a particular assignment of observations to the training set.  However, when we perform sample splitting, we typically randomly allocate datapoints to the training set. Under such a procedure, $I_{\bold{X}_{\mathrm{ss}}^{\mathrm{train}}(\theta)}$ from Example~\ref{theorem_noniid} becomes a random variable that depends on the particular split of the data. 
If all permissible random splits of the data are equally likely, then $
E\left[I_{\bold{X}_{\mathrm{ss}}^{\mathrm{train}}(\theta)}\right] = I_{\bold{X}^{\mathrm{train}}_{\mathrm{dt}}(\theta)}$, where the expected value is taken over all permissible splits of the data. The same equality holds for the test set, i.e. $
E\left[I_{\bold{X}_{\mathrm{ss}}^{\mathrm{test}}(\theta)}\right] = I_{\bold{X}^{\mathrm{test}}_{\mathrm{dt}}(\theta)}$. Despite this equivalence in the \emph{expected} information allocation, Proposition 1 from \cite{rasines2021splitting} provides clever insight that tells us why we might prefer the non-random information allocation of data thinning in this setting. Remark~\ref{remark_rasines_info} instantiates the general result of \cite{rasines2021splitting} to the simple setting of Example~\ref{theorem_noniid}. 

\begin{remark}[Effect on confidence interval width]
\label{remark_rasines_info}
In the context of Example~\ref{theorem_noniid}, suppose that our ultimate goal involves forming a confidence interval for $\beta$ using the test set. If we are using a maximum-likelihood estimator for $\beta$, then the width of the confidence interval computed on the test set under sample splitting and data thinning are proportional, respectively, to  $\frac{1}{I_{\bold{X}_{\mathrm{ss}}^\mathrm{test}}(\beta)}$ and  $\frac{1}{I_{\bold{X}_{\mathrm{dt}}^\mathrm{test}}(\beta)}$.  Jensen's inequality yields the following result:
$$
\E\left[ \frac{1}{I_{\bold{X}_{\mathrm{ss}}^\mathrm{test}}(\beta)} \right] \geq \frac{1}{\E\left[ I_{\bold{X}_{\mathrm{ss}}^\mathrm{test}}(\beta) \right] } = \frac{1}{I_{\bold{X}^{\mathrm{test}}_{\mathrm{dt}}(\beta)}}.
$$
Thus, the confidence intervals for $\beta$ computed using sample splitting will be wider, on average, than those computed using data thinning in the setting where our observations are not identically distributed. A natural corollary is that sample splitting will achieve lower power than data thinning in this setting, as we will see in Section~\ref{subsec:empirical_ss}. 
\end{remark}

\subsection{Empirical comparison to sample splitting}
\label{subsec:empirical_ss}

We now show empirically that sample splitting and data thinning achieve comparable performance in settings where both are applicable and where the observations are independent and identically distributed. 

We let $n=100$ and $p=20$, and we generate $10,000$ realizations of 
$\bold{Z} \in \mathbb{R}^{n \times p}$, where the entries of $\bold{Z}$ are drawn independently from $N(0,1)$. For each realization, we then generate $\bold{X} \mid \bold{Z} \sim  N_n \left( \bold{Z} \beta , I_n \right)$, where $\beta_1 = \beta_2 = \ldots = \beta_5 = \beta^*$ and $\beta_6 =\beta_7 =  \ldots = \beta_{20} = 0$. Our goal is to perform model selection to identify the covariates with non-zero entries in $\beta$, and then to form confidence intervals for the coefficients of the selected covariates. 

As we cannot naively use the entire dataset $\left(\boldsymbol{Z}, \boldsymbol{X}\right)$ to do both model selection and inference, we consider the following approach. 
\begin{list}{}{}
	\item[\textbf{Step 1:}] Split the data into a training set and a test set.
	\item[\textbf{Step 2:}] Perform forward stepwise regression on the training set to select a model that includes some subset of the $p=20$ covariates. We use the \texttt{R} function \texttt{step} with its default settings. 
	\item[\textbf{Step 3:}] Re-fit the selected model using the test set and report the standard confidence intervals for each selected coefficient.
\end{list}
We carry out this process using two different methods. 
\begin{list}{}{}
\item[\textbf{Sample splitting:}] Randomly generate a set $\mathrm{train} \subset \{1,\ldots,n\}$ with $| \mathrm{train} | = \epsilon n$. Let $\left\{ \left(\boldsymbol{Z}_i, X_i\right) : i \in \mathrm{train} \right\}$ be the training set and let $\left\{ \left(\boldsymbol{Z}_i, X_i\right) : i \in \{1,\ldots,n\} \setminus \mathrm{train} \right\}$  be the test set. 
\item[\textbf{Data thinning:}] Apply Algorithm~\ref{alg:datathin} to each $X_i$ for 	$i=1,\ldots,n$. Let $\left\{ \left(\bold{Z}_i, X_i^{(1)}\right) : i \in  \{1,\ldots,n\}  \right\} $ be the training set and let $\left\{ \left(\bold{Z}_i, X_i^{(2)}\right) : i \in  \{1,\ldots,n\}  \right\} $ be the test set.  	
\end{list}

We carry out each method using $\epsilon=0.2$ and $\epsilon=0.8$. For each method, we consider (i) \emph{detection}: the proportion of datasets for which $Z_3$ appears in the selected model, and (ii) \emph{power}: the proportion of datasets for which the confidence interval for $\beta_3$ does not include $0$, among those where $Z_3$ appeared in the selected model. These metrics focus on one of the important covariates, $Z_3$, but the results are similar for each of the important covariates (i.e. $Z_j$ for $j \in \{1,\ldots,5\}$).  The results are shown in the top row of Figure~\ref{figure_samplesplit}. As expected, data thinning and sample splitting achieve nearly identical results for these two metrics, with $\epsilon$ governing a tradeoff between detection and power. 

We then repeat the experiment, but we let $n=42$ and we let $\bold{Z}$ be a fixed matrix that contains a single row whose elements are drawn from $N(5,1)$ rather than $N(0,1)$. This corresponds to having a single high-leverage observation in the dataset that contains most of the information about $\beta$. 
We only consider $\epsilon \geq 0.5$, since in this setting where $n \approx 2 p$, using a smaller value of $\epsilon$ would complicate our ability to fit a linear regression model on the training set. As shown in the bottom row of Figure~\ref{figure_samplesplit}, data thinning outperforms sample splitting in this setting, because for any particular split, sample splitting either leaves very little information in the training set or very little in the test set. The power results confirm the insight from Remark~\ref{remark_rasines_info}. 

The general finding that data thinning outperforms sample splitting in this setting mirrors findings from  \cite{leiner2021data} and \cite{rasines2021splitting}, which is not surprising in light of Remark~\ref{remark:normal_equivalence}, which states that Gaussian data thinning proposal is equivalent (up to a simple rescaling) to the Gaussian randomization or fission proposals of these prior papers. Thus, the similar conclusions are not a coincidence, and the properties of data thinning seen in this setting can also be interpreted as properties of (Gaussian) data fission.  

\begin{figure}
\centering
\includegraphics[width=\textwidth]{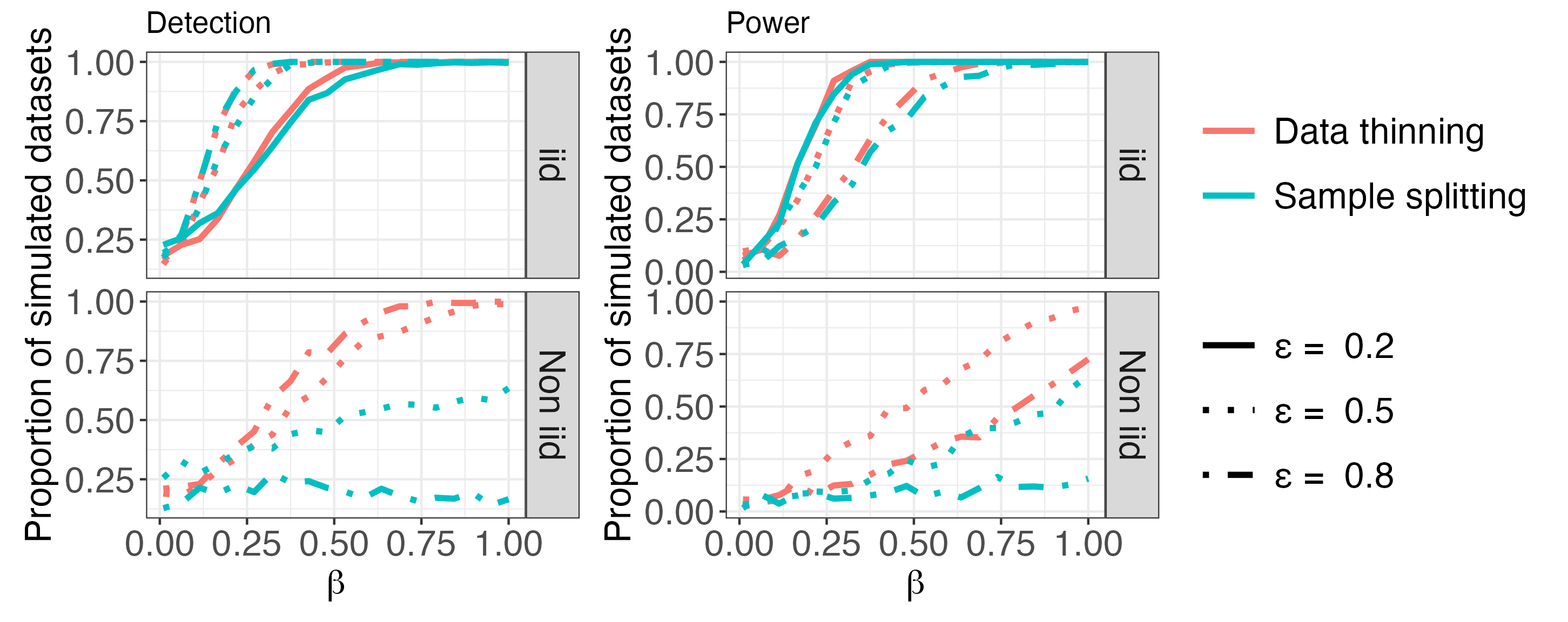}
\caption{Comparison of data thinning and sample splitting, using the detection and power metrics defined in Section~\ref{subsec:empirical_ss}. The top row shows the results of the large $n$ setting where the observations are independent and identically distributed (iid), and thus data thinning and sample splitting achieve nearly identical results across values of $\epsilon$. The bottom row shows the results in the small $n$ setting, in which the observations are not identically distributed (non iid) and the presence of a high leverage point causes data thinning to outperform sample splitting.}
\label{figure_samplesplit}	
\end{figure}

\section{Comparing data thinning and data fission} 
\label{appendix:leiner}

As mentioned in Section~\ref{sec1:intro}, the data fission proposal of \cite{leiner2021data} provides an alternate set of strategies to decompose a single realization $X$ into $\xo$ and $\xt$. In this section, we compare and contrast the two approaches. 

\subsection{Independent decompositions}
\label{appendix:leiner1}

\cite{leiner2021data} provide a strategy for obtaining independent $\xo$ and $\xt$ only in the case  where $X$ is Poisson or Gaussian. 

In the case of the Poisson distribution, the proposal of \cite{leiner2021data} coincides exactly with the proposal obtained from Algorithm~\ref{alg:datathin} in this paper. This proposal follows from a classical property of the Poisson distribution (see e.g. \citep{durrett2019probability}, Section 3.7.2), and has recently been applied in contexts related to that of this paper by \cite{oliveira2022coupled, sarkar2021separating, gerard2020data, chen2021estimating} and \cite{neufeld2022inference}.

In the case of the Gaussian distribution, the proposal of \cite{leiner2021data} has also been used by  \cite{tian2018selective}, \cite{oliveira2021unbiased}, and \cite{rasines2021splitting}, among others. It does not follow directly from Algorithm~\ref{alg:datathin} in this paper, since $X \neq \xo+\xt$. However, in Example~\ref{ex:leiner_normal}, we show the proposal of \cite{leiner2021data} is a simple rescaling of the proposal in this paper. 

\begin{example}[Comparison of two Gaussian decompositions]
Consider the task of splitting  the $\mathrm{N}(\mu, \sigma^2)$ distribution into two independent normally-distributed random variables, with $\sigma$ known.
The data thinning proposal is given in Table~\ref{tab:maintable}, and leads to $\xo \sim \mathrm{N}(\eps \mu, \eps \sigma^2)$ and $\xt \sim \mathrm{N}((1-\eps) \mu, (1-\eps) \sigma^2)$, where $\xo+\xt = X$. 

The data fission proposal is as follows: given a value of $\tau>0$, we draw $Z \sim \mathrm{N}(0, \sigma^2)$, and then let $\xo = X + \tau Z$  and  $\xt = X - \frac{1}{\tau} Z$.  Then,  $X' \sim \mathrm{N}(\mu, (1+\tau^2) \sigma^2)$  and  $X'' \sim \mathrm{N}(\mu, (1+\frac{1}{\tau^2}) \sigma^2)$, with $X' \perp\!\!\!\perp X''$. Under this decomposition, $X'+X'' \neq X$, but 
$$
\frac{1}{1+\tau^2} X' + \frac{\tau^2}{1+\tau^2} X'' = \frac{1}{1+\tau^2} (X + \tau Z) + \frac{\tau^2}{1+\tau^2} (X - \frac{1}{\tau} Z) = X.
$$
We can easily verify that, if we let $\eps = \frac{1}{1+\tau^2}$, then the random variables $\frac{1}{1+\tau^2} X'$ and $\frac{\tau^2}{1+\tau^2} X''$ obtained via data fission have the same distributions (both marginally and conditional on $X=x$) as $\xo$ and $\xt$ obtained via data thinning. For example, we can easily verify that $\frac{1}{1+\tau^2} X' \sim  \N\left( \frac{1}{1+\tau^2} \mu, \frac{1}{1+\tau^2} \sigma^2 \right) =  \N \left( \eps \mu, \eps \sigma^2 \right)$. Thus, the two decompositions are identical, up to a scaling of $\xo$ and $\xt$ by a (known) constant. 
\label{ex:leiner_normal}
\end{example}

The main idea of Example~\ref{ex:leiner_normal} extends to the decomposition of the multivariate normal given in Table~\ref{tab:maintable} of this paper, and the corresponding decomposition from \cite{leiner2021data}. 

\subsection{Non-independent decompositions}
\label{appendix:leiner2}

With the exception of the Gaussian and Poisson distributions, the decompositions of \cite{leiner2021data} do not yield $\xo$ and $\xt$ that are independent. Instead, the goal of \cite{leiner2021data} is to obtain a decomposition such that the distributions of $\xo$ and $\xt \mid \xo$ are tractable. While in principle we can fit a model to $\xo$ and validate it using the conditional distribution of $\xt \mid \xo$, we will see in this section that this can be difficult to carry out in practice. In particular, we note the following drawbacks of the non-independent decompositions of \cite{leiner2021data}. 

\begin{list}{}{}
\item[(1)] The distribution of $\xo$, and the conditional distribution of $\xt | \xo$, need not resemble the distribution of $X$. Thus, if the goal is to evaluate a potential model for $X$, it is not always clear what model to fit to $\xo$. We will illustrate this drawback in Example~\ref{ex:leiner-gam}.
\item[(2)] The parameters of interest are entangled in the conditional distribution of $\xt \mid \xo$. We will illustrate this issue in Example~\ref{ex:leiner_binom}.
\item[(3)]  The tuning parameter that governs the information trade-off between $\xo$ and $\xt$ can be hard to interpret.  For instance, in the case of the gamma decomposition in Example~\ref{ex:leiner-gam}, the tuning parameter is $B \in \{1,2,\ldots\}$, but in the case of the negative binomial distribution in Example~\ref{ex:leiner_binom}, it is $\eps \in (0,1)$. These both contrast with Example~\ref{ex:leiner_normal}, where the tuning parameter was $\tau > 0$. 
\item[(4)] The roles of $\xo$ and $\xt$ cannot be interchanged. For example, in the decomposition of the $\mathrm{Bernoulli}(\theta)$ distribution given in \cite{leiner2021data}, the distributions of $\xo$ and $\xt \mid \xo$ each contain information about $\theta$. However, the distribution of $\xo \mid \xt$ contains no information about $\theta$. Furthermore, while Remark 1 in \cite{leiner2021data} provides a strategy for obtaining multiple folds of training data for the data fission decompositions that are constructed using the ``conjugate prior" strategy, these folds are not marginally independent of one another (they are conditionally independent given $X$). Beyond these specific decompositions, \cite{leiner2021data} do not provide a clear strategy for extending their decompositions to the case of multiple folds. Thus, it is not clear in general how to use data fission decompositions to carry out cross validation. 
\end{list}
 
To illustrate point (1), we consider the gamma distribution. 

 \begin{example}[Gamma decomposition, data fission approach] \label{ex:leiner-gam}
Suppose $X \sim \mathrm{Gamma}(\alpha, \beta)$. For a tuning parameter $B \in \{1,2,\ldots \}$, \cite{leiner2021data} propose drawing $Z=(Z_1,\ldots,Z_B)$, where $Z_i \overset{\mathrm{ind.}}{\sim}  \mathrm{Poisson}(X)$, and thus each $Z_i$ marginally follows a $\mathrm{NegativeBinomial}(\alpha, 1/(\beta+1))$ distribution, and the $Z_i$ are independent conditional on $X$. Take $\xo=Z$, and $\xt = X$. Then, the conditional distribution of $\xt \mid \xo$ is $\mathrm{Gamma}(\alpha+\sum_{i=1}^B Z_i, \beta+B)$. 
\end{example}
This stands in notable contrast to Example~\ref{ex:gcs-gam-Kfold} in Section~\ref{sec1:intro}, in which data thinning provides independent (and gamma-distributed) random variables. Given that $\xo$ does not resemble $X$, it is not clear how to apply this decomposition in the setting of Example~\ref{subsubsec:gammasetup} from Section~\ref{sec:sim}, where we would like to apply a clustering algorithm to $\xo$ to estimate the true cluster structure of $X$. Unless $B=1$, $\xo$ and $X$ do not even have the same dimensions, which makes it difficult to know what type of clustering algorithm to apply or how to interpret the results. 

To illustrate point (2), we revisit Example~\ref{ap:Eval} to see a concrete example in which data thinning is straightforward but the proposal of \cite{leiner2021data}  is difficult to use in practice. 
 
\begin{example}[A comparison of negative binomial decompositions]
\label{ex:leiner_binom}
We observe \\ 
${X}_{ij} \sim  \mathrm{NegativeBinomial}\left( {r}_{ij} , {p}_{ij}\right)$ for $i=1,\ldots,n$ and $j=1,\ldots,p$. Our goal is to evaluate a function $\hat{\mu}(X)$ as an estimator for $\E[X]$, where $E[X]_{ij} = r_{ij}\frac{1-p_{ij}}{p_{ij}}$. 

From Table~\ref{tab:maintable}, data thinning requires ${r}_{ij}$ to be known, and yields ${X}^{(1)}_{ij} \sim  \mathrm{NegativeBinomial}\left(\eps {r}_{ij},{p}_{ij} \right)$ and ${X}^{(2)}_{ij} \sim  \mathrm{NegativeBinomial}\left((1-\eps) {r}_{ij}, {p}_{ij}\right)$, with ${X}^{(1)} \perp\!\!\!\perp {X}^{(2)}$, $\E[{X}^{(1)}] = \eps \E[{X}]$, and $\E[{X}^{(2)}] = (1-\eps) \E[{X}]$. Thus, as in Example~\ref{ap:Eval} and Example~\ref{ex:lossMSE}, $\hat{\mu}\left({X}^{(1)}\right)$ is an estimator for $\eps \E[{X}]$, and we can evaluate $\hat{\mu}({X}^{(1)})$ by computing the mean squared error between ${X}^{(2)}$ and $\frac{1-\eps}{\eps} \hat{\mu}\left({X}^{(1)}\right)$. 

For $\eps \in (0,1)$, the data fission proposal of \cite{leiner2021data} draws ${X}^{(1)}_{ij} \mid {X}_{ij} \sim \mathrm{Binomial}({X}_{ij},\eps)$ and sets ${X}^{(2)} = {X}-{X}^{(1)}$. Under this decomposition, \\
${X}^{(1)}_{ij} \sim \mathrm{NegativeBinomial}\left({r}_{ij}, \frac{{p}_{ij}}{{p}_{ij}+\eps(1-{p}_{ij})}\right)$, and so $\E[{X}^{(1)}]= \eps \E[{X}]$. Although marginally $\E[\xt] = (1-\epsilon) \E[X]$, 
as ${X}^{(1)}$ and ${X}^{(2)}$ are not independent, we cannot simply use mean squared error loss between ${X}^{(2)}$ and $\frac{1-\eps}{\eps} \hat{\mu}\left({X}^{(1)}\right)$ to evaluate the estimator. Instead, we must construct a loss function that evaluates $\hat{\mu}\left({X}^{(1)}\right)$ as an estimator of $\E[X]$ in the conditional distribution of $\xt \mid \xo$, which is given by
${X}^{(2)}_{ij} \mid {X}^{(1)}_{ij} \sim  \mathrm{NegativeBinomial}\left({r}_{ij}+{X}^{(1)}_{ij}, {p}_{ij} + \epsilon -{p}_{ij} \epsilon \right)$. As $E[{X}^{(2)}_{ij} \mid \xo_{ij}] = \left({r}_{ij} + {X}^{(1)}_{ij} \right) \left( \frac{1 - {p}_{ij} - \eps + \eps {p}_{ij}}{{p}_{ij} + \eps - \eps {p}_{ij}} \right)$ is not a simple function of $\E[X]$, this is not a straightforward task. At a minimum, it involves disentangling the roles of the parameters $r_{ij}$ and $p_{ij}$. 

Furthermore, while at first glance it might appear that an advantage of data fission over data thinning is that the former does not require knowledge of ${r}_{ij}$ to obtain ${X}_{ij}^{(1)}$ and ${X}_{ij}^{(2)}$, evaluating an estimator of $\E[X]$ using this conditional distribution will require knowing or accurately estimating the nuisance parameters ${r}_{ij}$ due to the aforementioned parameter entanglement issue. 
\end{example}

Similar issues arise for other decompositions given in \cite{leiner2021data}. For example, the data fission decomposition of the binomial distribution yields a complicated unnamed distribution for $\xt \mid \xo$, which would be very difficult to use in the context of Example~\ref{subsubsec:binsetup}.

\section{Simulation Study}
\label{sec:sim}

\subsection{Simulation setup}
\label{subsec:setup}

In this section, we focus on the application of data thinning to cross-validation in two settings. We contrast its performance to naive approaches that use the same data to both fit and validate the models. Specifically, we consider Examples~\ref{subsubsec:binsetup} and \ref{subsubsec:gammasetup}. In each of these examples, sample splitting is not a viable option, as the parameters of interest have dimensions equal to the number of observations. Furthermore, as pointed out in Section~\ref{appendix:leiner2}, applying data fission in these settings is not straightforward. 

\begin{example}[Choosing the number of principal components on binomial data]
\label{subsubsec:binsetup}
We generate data with $n=250$ observations and $d=100$ dimensions. Specifically, for $i=1,\dots,n$ and $j=1,\dots,d$, we generate $X_{ij} \overset{\mathrm{ind.}}{\sim}\text{Binomial}(r,p_{ij})$ where $r=100$ and $p$ is an unknown $n \times d$ matrix of probabilities. We construct $\text{logit}(p)$ as a rank-$K^*=10$ matrix with singular values $5,6,\dots,14$. Additional details are provided in Section~\ref{sec:simpar}. Our goal is to estimate $K^*$. 
\end{example}

\begin{example}[Choosing the number of clusters on gamma data]
\label{subsubsec:gammasetup}
We generate datasets $X \in \mathbb{R}^{n \times d}$ such that there are 100 observations in each of $K^*$ clusters, for a total of $n=100K^*$ observations. Our objective is to estimate $K^*$. We let $X_{ij} \overset{\mathrm{ind.}}{\sim}  \text{Gamma}(\lambda, \theta_{c_i,j})$, for $i=1,\ldots, n$ and $j=1,\ldots,d$, where  $c_i \in \{1,2,\dots,K^*\}$ indexes the true cluster membership of the $i$th observation. The shape parameter $\lambda$ is a known constant common across all clusters and all dimensions, whereas the rate parameter $\theta$ is an unknown $K^* \times d$ matrix such that each cluster has its own $d$-dimensional rate parameter. We generate data under two regimes: (1) a small $d$, small $K^*$ regime in which $d=2$ and $K^*=4$, and (2) a large $d$, large $K^*$ regime in which $d=100$ and $K^*=10$. The values of $\lambda$ and $\theta$ are provided in Section~\ref{sec:simpar}. A sample ``small $d$, small $K^*$" dataset is presented in Figure \ref{fig:2dgammasim}, alongside the output of data thinning with $\epsilon=0.5$.  
\end{example}

\begin{figure}
\centering
\includegraphics[width=0.8\linewidth]{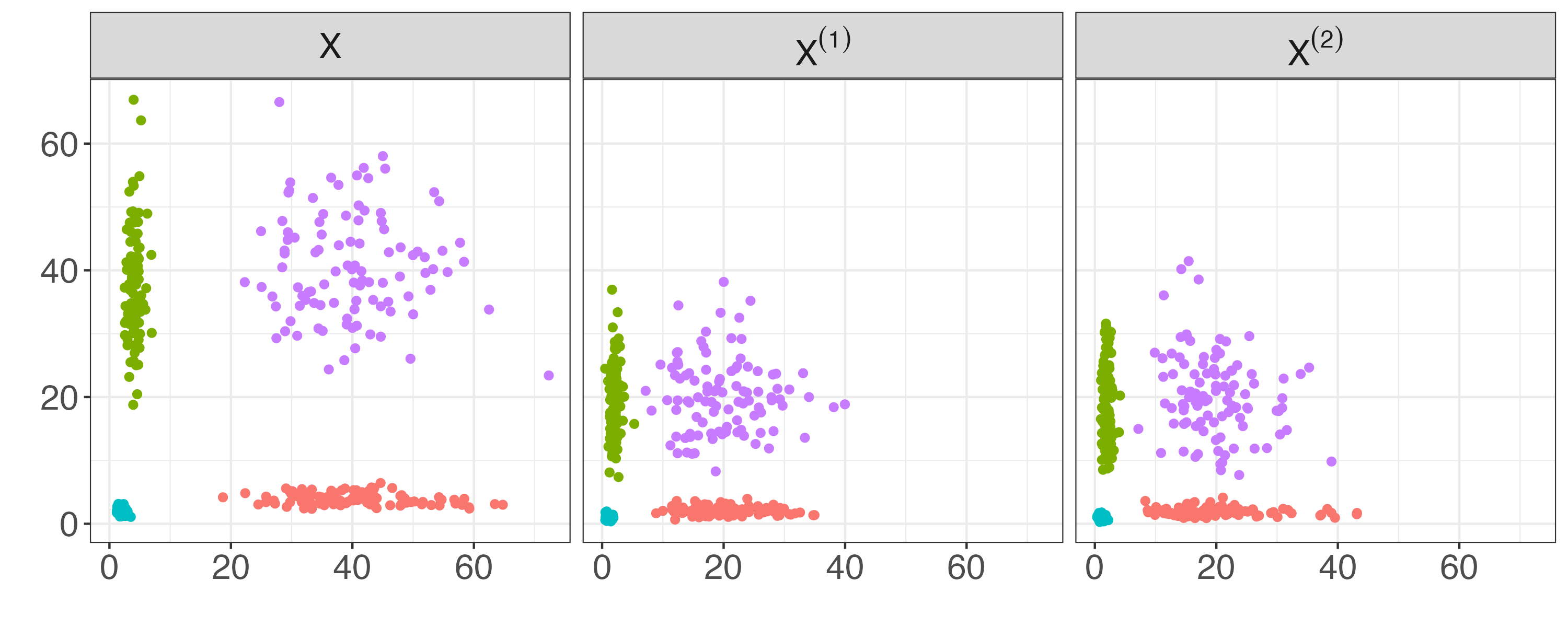}
\caption{\textit{Left}: A simulated dataset in the $d=2$, $K^*=4$ setting described in Example~\ref{subsubsec:gammasetup}. \textit{Center/Right}: The result of data thinning with $\epsilon=0.5$. }
\label{fig:2dgammasim}
\end{figure}

\subsection{Methods}
\label{subsec:methods}

We use Algorithm~\ref{alg:binpcaalg} to select the number of principal components in binomial data, as in Example~\ref{subsubsec:binsetup}, using data thinning. 

\begin{algorithm}[H]
\caption{Evaluating binomial principal components with negative log-likelihood loss}
\label{alg:binpcaalg}
\Input{A positive integer $K$, a matrix $X \in \mathbb{Z}_{[0,r]}^{n \times d}$, where $X_{ij} \overset{\textup{ind.}}{\sim}  \textup{Binomial}(r, p_{ij})$, and positive scalars $\epsilon^{\textup{(train)}}$ and $\epsilon^{\textup{(test)}} =1-\epsilon^{\textup{(train)}}$ such that  $\epsilon^{\textup{(train)}}r, \epsilon^{\textup{(test)}}r \in \mathbb{Z}_{> 0}$. } 
Apply data thinning to $X$ to obtain $\xtr$ and $\xte$, where $\xtr_{ij} \overset{\textup{ind.}}{\sim} \textup{Binomial}\left(\epsilon^{\textup{(train)}}r, p_{ij}\right)$ and  $\xte_{ij} \overset{\textup{ind.}}{\sim} \textup{Binomial}\left(\epsilon^{\textup{(test)}}r, p_{ij}\right)$. \\
Compute the singular value decomposition of the log-odds of $\xtr$, 
	$$
	\hat U \hat D \hat V^T = \textup{logit}\left\{(\xtr+0.001)/(\epsilon^{\textup{(train)}}r+0.002)\right\}.
	$$
	Pseudo-counts prevent taking the logit of 0 or 1. \\
Construct the rank-$K$ approximation of $\xtr$, $p^{(K)} := \textup{expit}\left(\hat U_{1:K} \hat D_{1:K} \hat V_{1:K}\right)$. \\
Compute the negative log-likelihood loss on $\xte$,
$
 -\sum_{i=1}^n \sum_{j=1}^d \log f\left(\xte_{ij} \Big\vert \epsilon^{\textup{(test)}}r, p^{(K)}_{ij}\right), 
$
where $f(\cdot \mid r,p)$ is the density function for the $\mathrm{Binomial}(r,p)$ distribution. 	
\end{algorithm}

We use Algorithm~\ref{alg:gammaclustalg} to select the number of clusters in gamma data, as in Example~\ref{subsubsec:gammasetup}, using data thinning.

\begin{algorithm}[H]
\caption{Evaluating gamma clusters with negative log-likelihood loss}	
\label{alg:gammaclustalg} 
\Input{A positive integer $K$, and $X \in \mathbb{R}_{> 0}^{n \times d}$ where $X_{ij} \overset{\textup{ind.}}{\sim} \textup{Gamma}\left(\lambda, \theta_{c_i,j}\right)$. Here, $\theta \in (0,\infty)^{K^* \times d}$ where $\theta_{c_i,j}$ is the true but unknown rate parameter for the $c_i$th cluster in the $j$th dimension, $c_i \in \{1,2,\dots,K^*\}$, and $\lambda$ is the known shape parameter. Also, positive scalars $\epsilon^{\textup{(train)}}$ and $\epsilon^{\textup{(test)}} =1-\epsilon^{\textup{(train)}}$. } 
Apply data thinning to $X$ to obtain $\xtr$ and $\xte$, where $\xtr_{ij} \overset{\textup{ind.}}{\sim} \textup{Gamma}\left(\epsilon^{\textup{(train)}}\lambda, \theta_{c_i,j}\right)$ and $\xte_{ij} \overset{\textup{ind.}}{\sim} \textup{Gamma}\left(\epsilon^{\textup{(test)}}\lambda, \theta_{c_i,j}\right)$. \\
 Run $K$-means on $\xtr$ to estimate $K$ clusters. Denote the cluster assignment of the $i$th observation as $\hat c_i$. \\
 Within each cluster, estimate the parameters using $\xtr$ \citep{gammaest,gammaest2}. Let $\hat\lambda^{(K)}$ and $\hat\theta^{(K)}$ denote the $K \times d$ estimated parameter matrices. \\
Compute the loss on $\xte$ as
	$
	-\sum_{i=1}^n \sum_{j=1}^d  \log f\left(\xte_{ij} \Big\vert \hat\lambda^{(K)}_{\hat c_i,j} \epsilon^{\textup{(test)}}/\epsilon^{\textup{(train)}} , \hat\theta^{(K)}_{\hat c_i,j}\right),
	$
where $f(\cdot \mid \lambda, \theta)$ is the density function for the  $\mathrm{Gamma}(\lambda,\theta)$ distribution. 
\end{algorithm}

We apply Algorithms \ref{alg:binpcaalg} and \ref{alg:gammaclustalg} in three different ways.
First, we apply them without modification, with $\epsilon^{\mathrm{(train)}} = 0.5$ and $\epsilon^{\mathrm{(train)}} = 0.8$.
Next, we slightly modify these algorithms by replacing step 1 with multi-fold thinning (Algorithm \ref{alg:datathin_manyfolds}) with $M=5$ and  $\epsilon_1=\dots=\epsilon_M=0.2$. For $m=1,\ldots,M$, we then perform steps 2--4 using $\xtr = X - X^{(m)}$, $\epsilon^{\textup{(train)}}=(M-1)/M$ and $\xte = X^{(m)}$, $\epsilon^{\textup{(test)}}=1/M$. We then average the loss functions obtained across the $M$ applications of step 4. Finally, we consider a naive method that re-uses data, by skipping step 1, and simply taking $\xtr=\xte=X$ in steps 2--4 and $\epsilon^{\textup{(train)}} = \epsilon^{\textup{(test)}}=1$ in step 4.

Our goal is to select the value of $K$ that minimizes the loss function. Because data thinning produces independent training and test sets, we expect that the data thinning approaches will produce U-shaped loss function curves, as a function of $K$. By contrast, in the naive approach, the full data $X$ is used to fit the model and to compute the loss functions in Algorithms~\ref{alg:binpcaalg} and ~\ref{alg:gammaclustalg}, resulting in monotonically decreasing loss curves, as a function of $K$. 

Other loss functions can be used in lieu of the negative log-likelihood loss in Algorithms \ref{alg:binpcaalg} and \ref{alg:gammaclustalg}. In Appendix~\ref{sec:simMSE}, we extend Algorithms \ref{alg:binpcaalg} and \ref{alg:gammaclustalg} to the case of mean squared error loss, and show similar results. 


\subsection{Results}

Figure~\ref{fig:role_NLL} displays the loss function for all three simulation settings as a function of $K$; results have been averaged over $2,000$ simulated datasets and rescaled to the $[0,1]$ interval for ease of comparison. The values of $K$ with the lowest average loss function are circled on the plots. As expected, the data thinning approaches in Figure \ref{fig:role_NLL} exhibit sharp minimum values, as opposed to the monotonically decreasing 
curves produced by the naive method. The data thinning approaches correctly select the true value of $K=K^*$ in all three settings, except for data thinning with $\epsilon^\mathrm{(train)}=0.5$ in the binomial principal components setting. In that case, the low value of $\epsilon^\mathrm{(train)}$ allocates too much information to the test set, resulting in inadequate signal from the weakest principal components in the training set. Selecting a larger value of $\epsilon^\mathrm{(train)}$ remedies this issue, as seen with $\epsilon^\mathrm{(train)}=0.8$.

\begin{figure}
\centering
\includegraphics[width=0.8\textwidth]{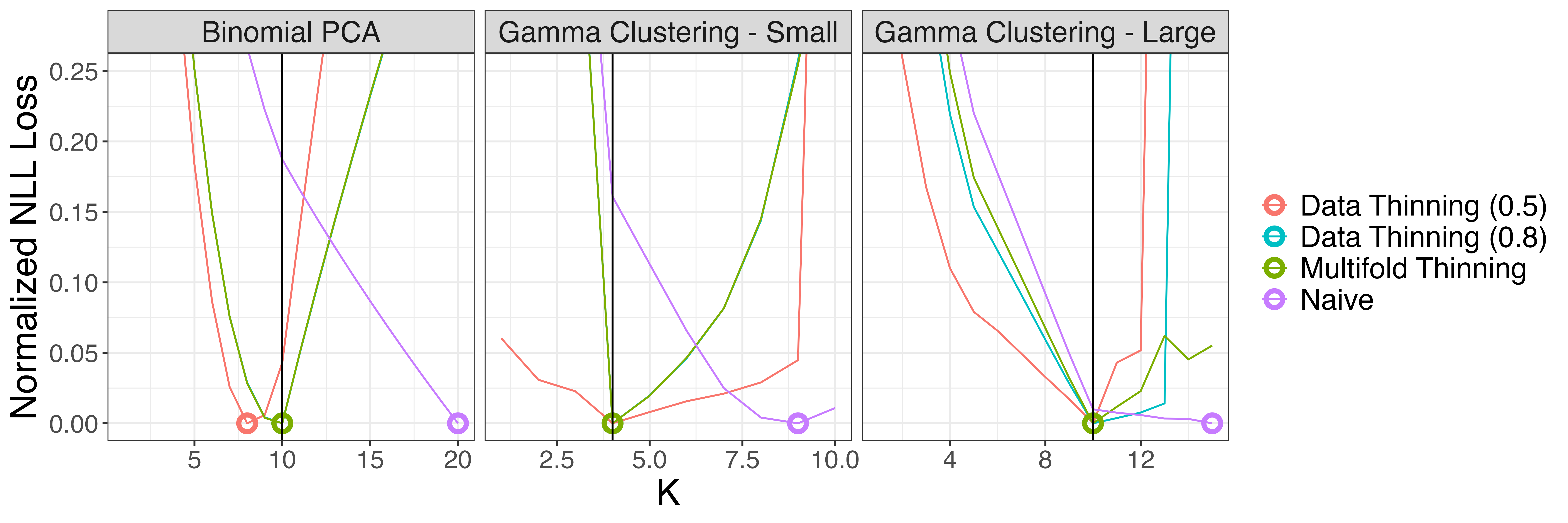}
\caption{The negative log-likelihood loss averaged over 2,000 simulated data sets, as a function of $K$, for the naive method (purple), data thinning with $\epsilon^\mathrm{(train)}=0.5$ (red), data thinning with $\epsilon^\mathrm{(train)} = 0.8$ (blue), and multifold thinning with $M=5$ folds (green). Each curve has been rescaled to take on values between $0$ and $1$, for ease of comparison. The minimum loss values for each method are circled, and $K^*$ is indicated by the vertical black line.}
\label{fig:role_NLL}
\end{figure}

We further investigate the role of $\epsilon^\mathrm{(train)}$ by repeating the simulation study using different values of $\epsilon^\mathrm{(train)}$ for single-fold data thinning. In Figure \ref{fig:role_eps}, we plot the proportion of simulations that select the correct value of $K^*$ (i.e. the proportion of simulations in which the loss function is minimized at $K=K^*$) in each of the three settings, as a function of $\epsilon^\mathrm{(train)}$. We find that in the gamma clustering simulations, lower values of $\epsilon^\mathrm{(train)}$ are adequate. However, settings with weaker signal, such as the binomial principal components example, require larger values of $\epsilon^\mathrm{(train)}$ to identify the true latent structure. In all settings, as $\epsilon^\mathrm{(train)}$ approaches $1$, performance begins to decay. This is a consequence of inadequate information remaining in the test set under large values of $\epsilon^\mathrm{(train)}$, and is consistent with the discussion of Section \ref{subsec:eps}. These findings suggest that in practice, the optimal value of $\epsilon^\mathrm{(train)}$ is context-dependent. 

\begin{figure}
\centering
\includegraphics[width=0.8\textwidth]{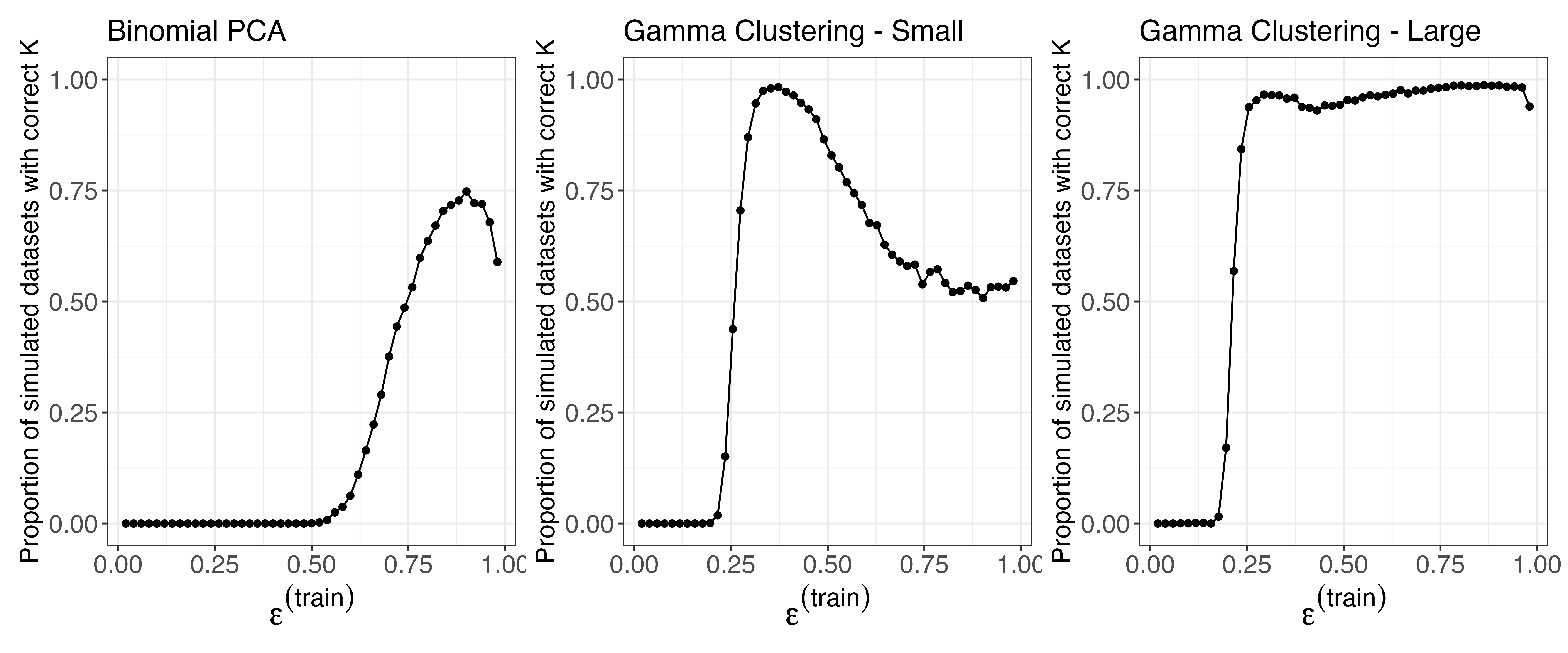}
\caption{The proportion of simulations for which data thinning selects the true value of $K^*$ with the negative log-likelihood loss, as a function of $\epsilon^\mathrm{(train)}$, for the simulation study described in Section \ref{subsec:setup}. The optimal value of $\epsilon^\mathrm{(train)}$ depends on the problem at hand.}	
\label{fig:role_eps}
\end{figure}

Finally, we examine the benefits of multifold data thinning over single-fold data thinning. Figure \ref{fig:role_folds} displays histograms of the number of simulations that select each value of $K$. Here we only include data thinning with $\epsilon^\mathrm{(train)}=0.8$ and multifold thinning with $M=5$, so that both methods use the same allocation of information between training and test sets. We see that multifold thinning generally selects the correct value of $K$ more often than single-fold data thinning, mirroring the improvement of $M$-fold cross-validation using sample splitting over single-fold sample splitting in supervised settings. However, in the large gamma setting, the signal is strong enough that multifold thinning does not provide a benefit over single-fold thinning.

\begin{figure}
\centering
\includegraphics[width=0.8\textwidth]{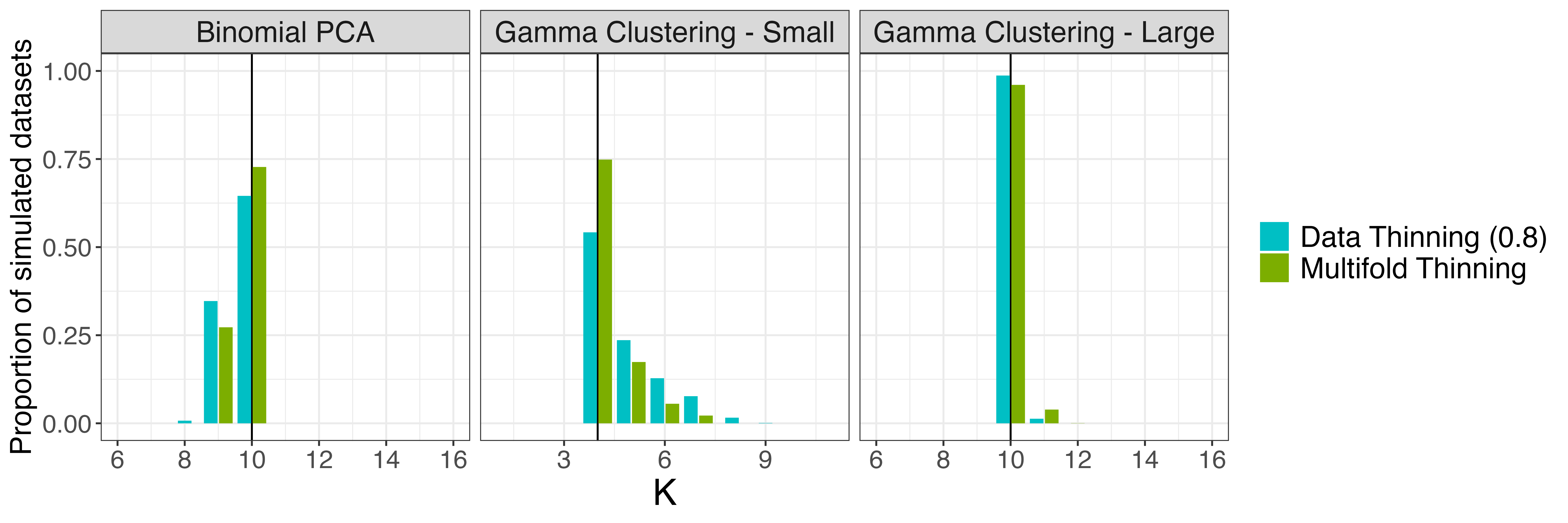}
\caption{The proportion of simulated data sets in which each candidate value of $K$ is selected, with the negative log-likelihood loss, under data thinning with $\epsilon^\mathrm{(train)}=0.8$ (blue) and multifold thinning with $M=5$ (green), for each of the simulation settings described in Section \ref{subsec:setup}. The true value of $K^*$ is indicated by the vertical black line. Multifold thinning tends to select the true value of $K$ more often than single-fold thinning.}	
\label{fig:role_folds}
\end{figure}

\section{Selecting the number of principal components in gene expression data}
\label{sec:data}

In this section, we revisit an analysis of a dataset from a single-cell RNA sequencing experiment conducted on a set of peripheral blood mononuclear cells. The dataset is freely available from 10X Genomics, and was previously analyzed in the ``Guided Clustering Tutorial" vignette  \citep{seuratvignette} for the popular \texttt{R} package \texttt{Seurat}  
\citep{seurat1, seurat2, seurat4}.

The dataset $X$ is a 
sparse matrix of non-negative integers, representing counts from  $32,738$ genes in each of $2,700$ cells. We consider applying principal components analysis to learn a low-dimensional representation of the data. In the \texttt{Seurat} vignette, filtering, normalization, log-transformation, feature selection, centering, and scaling are applied to the data, yielding a transformed matrix $\tilde{Y} \in \mathbb{R}^{2638 \times 2000}$. Details are provided in Section~\ref{appendix:seurat}. Finally, the singular value decomposition of $\tilde{Y}$ is computed, such that $\tilde{Y} = UDV^T$. Here we let $U_k$ represent the $k$th column of the matrix $U$, and let $U_{1:K} D_{1:K} V_{1:K}^T$ represent the rank-$K$ approximation of $\tilde{Y}$. 

Our goal is to select the number of dimensions to use in this low-rank approximation. In the \texttt{Seurat} vignette, the authors rely on heuristic solutions such as looking for an elbow in the plot of the standard deviation of $U_K D_K$ as a function of $K$; see Figure~\ref{fig:Seuratplot}(a) \citep{james2013introduction}. 
Based on the elbow plot, the authors suggest retaining around $7$ principal components. Other heuristic approaches suggest as many as $12$ principal components.

Before introducing the data thinning solution, we introduce a squared-error based formulation that is mathematically equivalent to the traditional elbow plot (see Section~\ref{appendix:seurat}), but will facilitate a direct comparison with data thinning. For $K=1,\ldots,20$, we compute the sum of squared errors between the matrix $\tilde{Y}$ and its rank-$K$ approximation:  
$$
\left\| \tilde{Y} - U_{1:K} D_{1:K} V_{1:K}^T\right\|_F^2.
$$
Because the low-rank approximation $U_{1:K} D_{1:K} V_{1:K}^T$ is computed using $\tilde{Y}$, this loss function monotonically decreases with $K$. A heuristic solution for deciding how many principal components to retain involves looking for the point in which the slope of the curve in Figure~\ref{fig:Seuratplot}(b) begins to flatten. While this appears to happen around 5--7 principal components, which is consistent with the finding from Figure~\ref{fig:Seuratplot}(a), the exact number of principal components to retain is still unclear. We now 
show that data thinning provides a principled approach for estimating the number of principal components. 

Single-cell RNA-sequencing data are often modeled as independent Poisson random variables \citep{wang2018gene, sarkar2021separating}. Thus, we assume that $X_{ij} \sim \mathrm{Poisson}(\Lambda_{ij})$. Starting with the raw data matrix $X \in \mathbb{Z}_{\geq 0}^{2700 \times 32738}$, we perform Poisson data thinning with $\epsilon=0.5$ to obtain a training set $X^{(1)}$ and a test set $X^{(2)}$, which are independent if the Poisson assumption holds. Furthermore, as $\eps = 0.5$, they are identically distributed. We then carry out the data processing described in Section~\ref{appendix:seurat} on $\xo$ to obtain $\tilde{Y}^{(1)} \in \mathbb{R}^{2638 \times 2000}$. We obtain $\tilde{Y}^{(2)}\in \mathbb{R}^{2638 \times 2000}$ by applying the same data processing steps to $\xt$, but retaining only the features that were selected on $\xo$, so that the rows and columns of $\tilde{Y}^{(1)}$ and $\tilde{Y}^{(2)}$ correspond to the same genes and cells. Details are in Section~\ref{appendix:seurat}. 

We compute the singular value decomposition on the training set, $\tilde{Y}^{(1)} = U^{(1)} D^{(1)} (V^{(1)})^T$. For a range of values of $K$, we then compute the sum of squared errors between $\tilde{Y}^{(2)}$ and $U_{1:K}^{(1)} D_{1:K}^{(1)} (V_{1:K}^{(1)})^T$: 
\begin{equation}
\label{eq:seurat_datathin_sse}
\left\| \tilde{Y}^{(2)} - U_{1:K}^{(1)} D_{1:K}^{(1)} (V_{1:K}^{(1)})^T\right\|_F^2.
\end{equation}
The results are shown in Figure~\ref{fig:Seuratplot}(c). As we are not computing and evaluating the singular value decomposition using the same data, the plot of $K$ vs. the loss function is not monotonically decreasing in $K$. Instead, it reaches a clear minimum at $K=7$, suggesting that the rank-$7$ approximation provides the best fit to the observed data. Thus, data thinning provides a simple and non-heuristic way to select the number of principal components. 

\begin{figure}
\centering
\includegraphics[width=\textwidth]{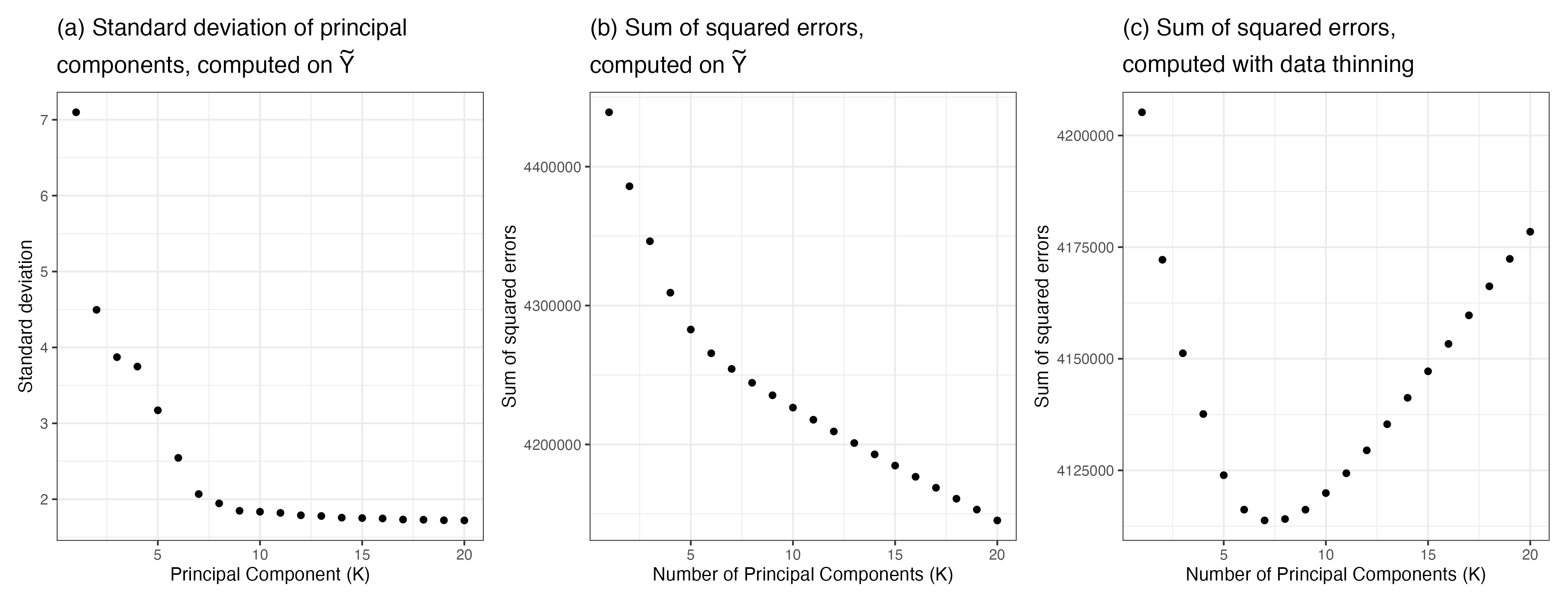}
\caption{
Results for the data analysis in Section~\ref{sec:data}.
\emph{(a)} An ``elbow plot" of the standard deviation of the principal components, which reproduces the plot given in the Seurat guided clustering tutorial. 
\emph{(b)} Due to the relationship between the sum of squared errors and the standard deviation of the principal components (see Section~\ref{appendix:seurat}), looking for an elbow in (a) is equivalent to looking for an elbow in (b). 
\emph{(c)} The data-thinning version of (b), which shows a clear minimum in the loss function at 7 principal components.}
\label{fig:Seuratplot}	
\end{figure}

\begin{remark}[Choice of $\eps$]
If we had chosen $\eps \neq 0.5$, then while $\E[\xt] = \frac{1-\eps}{\eps} \E[\xo]$, the relationship between $\E[\tilde{Y}^{(1)}]$ and $\E[\tilde{Y}^{(2)}]$ would depend on the details of the data processing described in Section~\ref{ap:Eval}, and the loss function in \eqref{eq:seurat_datathin_sse} would need to be modified accordingly. 
\end{remark}

\begin{remark}[Overdispersion]
While we used a Poisson model for scRNA-seq data, there is evidence that a negative binomial model may be preferable in some settings. It is possible to modify the analysis in this section using negative binomial data thinning, as in \cite{neufeld2023negative}.  
\end{remark}

\section{Discussion}
\label{sec:discussion}

We have proposed data thinning, a new way to decompose a single observation into two or more independent observations that sum to yield the original. This proposal applies to a very broad class of distributions. Furthermore, we have compared data thinning to sample splitting, and have shown that the former is applicable in many cases that the latter is not, and may be preferable even when both are applicable. 

In Section~\ref{subsec:nuissance}, we considered the impact of using the incorrect value of a nuisance parameter when performing data thinning, but we did not consider what happens when the nuisance parameter is estimated using the data itself. In future work, we will consider the theoretical and empirical implications of performing data thinning with an estimated nuisance parameter. Furthermore, we focused on convolution-closed distributions and thus used additive decompositions where $X = \xo+\xt$. For distributions with bounded support, such as the beta distribution, non-additive decompositions are needed. We leave such decompositions to future work. 

An R package implementing data thinning and scripts to reproduce the results in this paper are available at {https://anna-neufeld.github.io/datathin/}.

\section{Acknowledgements}

This material is based upon research supported in part by the Office of Naval Research (award number N000142312589), the National Science Foundation (award Number 2322920), the Simons Foundation (Simons Investigator Award in Mathematical Modeling of Living Systems), the National Institutes of Health (R01 EB026908, R01 GM123993, and R01 DA047869), and the Keck Foundation. Lucy Gao and Ameer Dharamshi were supported in part by the Natural Sciences and Engineering Research Council of Canada.  

\appendix

\section{Proofs from Section~\ref{sec:proposal}}

\subsection{Proof of Theorem~\ref{theorem:datathin}}
\label{appendix:mainproof}

This proof is due to \cite{joe1996time} and \cite{jorgensen1998stationary}, but has been adapted to fit our notation. 

Let $x$, our observed data, be a realization of a random variable $X \sim F_\lambda$. Let $\eps \in (0,1)$ be chosen such that $\eps \lambda$ and $(1-\eps) \lambda$ are in the parameter space $\Lambda$. 
We draw $\xo \mid X=x \sim G_{\eps \lambda, (1-\eps) \lambda, x}$, where this notation was defined in Section~\ref{subsec:convclosed}, and let 
$\xt = X - \xo$. 

Separately, let $X' \sim F_{\epsilon \lambda}$ and $X'' \sim F_{(1-\epsilon) \lambda}$ be independent, and let $Y = X' + X''$. As $F_\lambda$ is a convolution-closed distribution, it follows that $Y \sim F_\lambda$ and thus we know that $Y$ has the same marginal distribution as $X$. 

We first argue that the joint distribution of $(\xo, X)$ is the same as the joint distribution of $(X', Y)$. The conditional distribution of $\xo \mid X$ is the same as the conditional distribution of $X' \mid Y$ by definition of the distribution $G_{\eps \lambda, (1-\eps) \lambda, x}$. Furthermore, we have already noted that the marginal distributions of $X$ and $Y$ are the same. Thus, the joint distribution of $(\xo, X)$ is the same as the joint distribution of $(X', Y)$.

As $\xt$ is deterministic given $\xo$ and $X$, the joint distribution of $(\xo, \xt)$ is the same as the joint distribution of $(\xo, X)$. Similarly, the joint distribution of $(X',Y)$ is the same as the joint distribution of $(X', X'')$. Thus, the joint distribution of $(\xo,\xt)$ is the same as the joint distribution of $(X', X'')$. As the joint distribution of $X'$ and $X''$ is known to be the product of independent distributions $F_{\epsilon \lambda}$ and 
$F_{(1-\epsilon) \lambda}$, this completes the proof of parts (i) and (ii) of Theorem~\ref{theorem:datathin}.

The final statement of Theorem~\ref{theorem:datathin} follows directly from Definition~\ref{def:linExp}.

\subsection{Proof of Proposition~\ref{prop:normalnuissance}}
\label{appendix:nuissanceproof}

To prove (i), note that since $\xo \mid X=x$ is normally distributed and $X$ is normally distributed, a well-known property of the normal distribution tells us that the marginal distribution of $\xo$ is normal. We then use the law of total expectation and the law of total variance to compute its mean and variance.
\begin{align*}
\E[\xo] &= \E[\E[\xo \mid X]] = \E[\eps X] = \eps \mu \\
\Var(\xo) &= \Var\left(\E\left[\xo \mid X\right]\right) + \E\left[ \Var \left( \xo \mid X \right) \right]  \\
&= \Var\left( \eps X \right) + \E \left( \eps (1-\eps) \tilde{\sigma}^2 \right) \\
&= \epsilon^2 \sigma^2 + \eps(1-\eps)\tilde{\sigma}^2.
\end{align*}
To prove (ii), note that the difference between two normally distributed variables ($X$ and $\xo$) is normal. Then note that
 \begin{align*}
\E[\xt] &= \E[X]-\E[\xo] = \mu - \eps \mu = (1-\eps) \mu.\\
\Var(\xt) &=  \Var\left(\E\left[\xt \mid X\right]\right) + \E\left[ \Var \left( \xt \mid X \right) \right] \\
&=  \Var\left(\E\left[X - \xo \mid X\right]\right) + \E\left[ \Var \left( X - \xo \mid X \right) \right] \\
&=  \Var\left((1-\eps) X \right) + \E\left[ \Var \left( \xo \mid X \right) \right] \\
&=  (1-\eps)^2 \sigma^2  + \eps (1-\eps) \tilde{\sigma}^2,
\end{align*}
which completes the proof of (ii). Finally, to prove (iii), note that
\begin{align*}
2  \Cov(\xo, \xt) &= \Var(X) - \Var(\xo) - \Var(\xt) \\
 &= 	\sigma^2 - \epsilon^2 \sigma^2 - \eps(1-\eps)\tilde{\sigma}^2 - (1-\eps)^2 \sigma^2 - \eps(1-\eps)\tilde{\sigma}^2 \\
&= 2 \eps(1 - \epsilon) \left(\sigma^2 - \tilde{\sigma}^2 \right).
\end{align*}

\subsection{Proof of Proposition~\ref{prop:nbnuissance}}
\label{appendix:nuissanceproof_nb}

Recall that if $A \sim \mathrm{BetaBinomial}\left(r, \alpha, \beta \right)$, then $\E[A] = \frac{r\alpha}{\alpha+\beta}$ and $\Var(A) = \frac{r\alpha\beta(\alpha+\beta+r)}{(\alpha+\beta)^2(\alpha+\beta+1)}$. Then we can derive the marginal variance of $\xo$ using the law of total variance and the fact that $\xo \mid X \sim \mathrm{BetaBinomial}(X, \eps \tilde{r}, (1-\eps) \tilde{r})$. 
\begin{align*}
\Var(\xo) &= \E[\Var(\xo \mid X)] + \Var(\E[\xo \mid X]) \\
&= \E \left[ \frac{X \epsilon (1-\epsilon)  (\tilde{r}+X)}{(\tilde{r}+1)}\right] + \Var \left( \eps X \right)  \\
&= \frac{ \epsilon (1-\epsilon)}{\tilde{r}+1} \left( \tilde{r}\E\left[ X \right] + \E\left[ X^2 \right]\right) + \epsilon^2 \Var\left( X \right) \\
&= \frac{ \epsilon (1-\epsilon)}{\tilde{r}+1} \left( \tilde{r} \E\left[ X \right] + \Var(X) + E[X]^2 \right) + \epsilon^2 \Var\left( X \right). 
\end{align*}
Next note that $\Var(\xt \mid X) = \Var(X - \xo \mid X) =\Var(\xo \mid X)$ and that $\E[\xt \mid X] = (1-\eps) X$. Thus, we arrive at:
\begin{align*}
\Var(\xt) &= \frac{ \epsilon (1-\epsilon)}{\tilde{r}+1} \left( \tilde{r} \E\left[ X \right] + \Var(X) + E[X]^2 \right) + (1-\epsilon)^2 \Var\left( X \right). 
\end{align*}
To derive the covariance, note that:
\begin{align*}
2 \Cov(\xo, \xt) &= \Var(X) - \Var(\xo) - \Var(\xt) \\
&= \Var(X) - 2 \frac{ \epsilon (1-\epsilon)}{\tilde{r}+1} \left( \tilde{r} \E\left[ X \right] + \Var(X) +  E[X]^2 \right) - \left(\epsilon^2 + (1 -  \epsilon)^2 \right) \Var\left( X \right)  \\
&= - 2 \frac{ \epsilon (1-\epsilon)}{\tilde{r}+1} \left( \tilde{r} r \frac{1-p}{p} + r \frac{1-p}{p^2} + r^2  \frac{(1-p)^2}{p^2}\right) +  2 \epsilon (1 - \epsilon) r \frac{1-p}{p^2}  \\
&= 2 \epsilon (1 - \epsilon) r \frac{(1-p)^2}{p^2} \left(1 - \frac{r+1}{\tilde{r}+1} \right).
 \end{align*}

\subsection{Proof of Proposition~\ref{prop:gamma_nuissance}}
\label{appendix:nuissanceproof_gamma}

The proof structure is identical to those of Proposition~\ref{prop:normalnuissance} and Proposition~\ref{prop:nbnuissance}. 
We start by recalling that if $A \sim \mathrm{Gamma}(\alpha,\beta)$, then $\E[X] = \frac{\alpha}{\beta}$ and $\Var(X) = \frac{\alpha}{\beta^2}$. Then:
\begin{align*}
\Var(\xo) &= E\left[\Var\left(\xo \mid X\right)\right] + \Var\left( \E\left[ \xo \mid X\right]\right)	 \\
&= E\left[\Var\left(X Z \mid X\right)\right] + \Var\left( \E\left[ X Z \mid X\right]\right) \\
&= E\left[X^2 \Var\left(Z \right)\right] + \Var\left( X \E\left[ Z \right]\right) \\
&= E\left[X^2 \left( \frac{\eps(1-\eps)}{\tilde{\alpha}+1}\right)\right] + \Var\left( X \eps \right)  \\
&=  \frac{\eps(1-\eps)}{\tilde{\alpha}+1} \left( \Var(X) + \E[X]^2 \right) + \epsilon^2 \Var\left( X \right). 
\end{align*}
Similarly, we note that $\Var(\xt \mid X) = \Var(X - \xo \mid X) = \Var(\xo \mid X)$, while $\E(\xt \mid X) = (1-\eps) X$. This allows us to do a similar derivation and arrive at: 
\begin{align*}
\Var(\xt) &= E\left[\Var\left(\xt \mid X\right)\right] + \Var\left( \E\left[ \xt \mid X\right]\right)	 \\
&=  \frac{\eps(1-\eps)}{\tilde{\alpha}+1} \left( \Var(X) + \E[X]^2 \right) + (1-\epsilon)^2 \Var\left( X \right).
\end{align*}
Finally,
\begin{align*}
2 \Cov(\xo, \xt) &= \Var(X) - \Var(\xo) - \Var(\xt) \\
&=\Var(X) - 2  \frac{\eps(1-\eps)}{\tilde{\alpha}+1} \left( \Var(X) + \E[X]^2 \right) - \left( \epsilon^2 + 1 - 2 \eps + \epsilon^2\right) \Var(X) \\
&= - 2  \frac{\eps(1-\eps)}{\tilde{\alpha}+1} \left( \Var(X) + \E[X]^2 \right) + 2 \epsilon(1 - \eps) \Var(X) \\
&= -  2 \frac{\eps(1-\eps)}{\tilde{\alpha}+1} \left( \frac{\alpha (1+\alpha)}{\beta^2} \right) + 2 \epsilon(1 - \eps) \frac{\alpha}{\beta^2} \\
&= 2 \eps (1-\eps) \frac{\alpha}{\beta^2} \left( 1 - \frac{\alpha+1}{\tilde{\alpha}+1}\right).
\end{align*}

\section{Proof of Theorem~\ref{theorem:datathin_manyfolds}}
\label{appendix:mainproof_multi}

The proof is nearly identical to that of Theorem~\ref{theorem:datathin}. It extends ideas from 
 \cite{jorgensen1998stationary} and \cite{joe1996time} to the setting of multiple folds.
 
 Let $x$, our observed data, be a realization of random variable $X \sim F_\lambda$. Let $\eps_1,\ldots, \eps_M$ be chosen such that $\sum_{m=1}^M \eps_m = 1$, $\eps_m > 0$, and $\eps_m \lambda$ is in the parameter space $\Lambda$ for $m = 1,\ldots,M$. 
 
 Suppose we draw $\left(\xo, \ldots, \xM\right) \mid X=x \sim G_{\eps_1 \lambda, \eps_2 \lambda, \ldots, \eps_M \lambda, x}$, where $G_{\eps_1 \lambda, \eps_2 \lambda, \ldots, \eps_M \lambda, x}$ was defined in Section~\ref{sec:multiplefolds}.

Separately, let $X_1, X_2, \ldots, X_M$ be mutually independent random variables, where $X_m \sim F_{\eps_m \lambda}$, and let $Y = \sum_{m=1}^M X_m$. As $F_\lambda$ is a convolution-closed distribution, we know that $Y \sim F_\lambda$ and thus $Y$ has the same marginal distribution as $X$. 

The conditional joint distribution of $(\xo, \xt, \ldots, \xM) \mid X$ is the same as the conditional joint distribution of $(X_1, X_2, \ldots, X_M) \mid Y$, by definition of the distribution $G_{\eps_1 \lambda,\ldots, \eps_M \lambda, x}$. Furthermore, we have already seen that the marginal distributions of $X$ and $Y$ are the same. Thus, the marginal joint distribution of $(\xo, \xt, \ldots, \xM)$ is the same as the marginal joint distribution of $(X_1, X_2, \ldots, X_M)$. Furthermore, by construction, the marginal joint distribution of $(X_1, X_2, \ldots, X_M)$ is that of $M$ mutually independent random variables, where $X_m \sim F_{\eps_m \lambda}$. This concludes the proof of Theorem~\ref{theorem:datathin_manyfolds}, parts 1-3. Part 4 of Theorem~\ref{theorem:datathin_manyfolds} follows directly from Definition~\ref{def:linExp}.

\section{Proof of Theorem~\ref{theorem_epsilon}}
\label{appendix:eps_proofs}

To prove Theorem~\ref{theorem_epsilon}, we rely on Lemma~\ref{lemma1} and Lemma~\ref{lemma2}. 

\begin{lemma}
\label{lemma1}
Suppose that we thin a random variable $X \sim F_{\lambda}$ using Algorithm~\ref{alg:datathin_manyfolds} with $\epsilon_1 = \ldots = \epsilon_M = \frac{1}{M}$ to obtain $\xo, \ldots, \xM$. Let $I_X(\theta)$ denote the Fisher information contained in $X$ about an unknown parameter $\theta$ (assume that this Fisher information exists). Then the Fisher information contained in $X^{(m)}$ for $m=1,\ldots,M$ about $\theta$, denoted $I_{X^{(m)}}(\theta)$, is equal to $\frac{1}{M}I_X(\theta)$. 
\end{lemma}

\begin{proof}
We began with a random variable $X$ and we constructed $\left( \xo,\ldots,\xm \right)$ without knowledge of $\theta$. We cannot create information from nothing. Thus, if $I_{\left(\xo,\ldots,\xm\right)}(\theta)$ denotes the information about $\theta$ in the joint distribution of $\left(\xo,\ldots,\xm\right)$, then
\begin{equation}
\label{ineq1}
I_{\left(\xo,\ldots,\xm\right)}(\theta) \leq I_X(\theta).
\end{equation}
To prove \eqref{ineq1} formally, we can use the chain rule property of Fisher information to write $I_{\left(\xo,\ldots,\xm, X \right)}(\theta)$ in two ways: 
\begin{equation*}
\label{chainrule}
I_{\left(\xo,\ldots,\xm, X \right)}(\theta)  = I_{\left(\xo,\ldots,\xm \right)}(\theta) + I_{X \mid \left(\xo,\ldots,\xm \right)} (\theta) = I_{\left(\xo,\ldots,\xm\right) \mid X}(\theta) + I_{X} (\theta).  
\end{equation*}
Since $\theta$ is unknown during the thinning process, the distribution $\left(\xo,\ldots,\xm\right) \mid X$ must not involve $\theta$, so $I_{\left(\xo,\ldots,\xm\right) \mid X}(\theta) = 0$. This implies that
\begin{equation*}
\label{chainrule2}
I_{\left(\xo,\ldots,\xm \right)}(\theta) + I_{X \mid \left(\xo,\ldots,\xm \right)} (\theta) = I_{X} (\theta),
\end{equation*}
and since Fisher information is non-negative, \eqref{ineq1} follows.

Separately, since $X = \xo + \ldots + \xm$, we can reconstruct $X$ from $\left(\xo,\ldots,\xm\right)$, and so a fundamental property of Fisher information tells us that
\begin{equation}
\label{ineq2}
I_X(\theta) = I_{\sum_{m=1}^M \xm} (\theta) \leq I_{\left(\xo,\ldots,\xm\right)}(\theta). 
\end{equation}
Combining \eqref{ineq1} and \eqref{ineq2}, we have that
\begin{equation*}
\label{equality}
I_X(\theta) = I_{\left(\xo,\ldots,\xm\right)}(\theta). 
\end{equation*} 
Finally, as $\xo,\ldots,\xM$ are independent and identically distributed,
\begin{equation*}
\label{final}
I_X(\theta) = I_{\left(\xo,\ldots,\xM\right)}(\theta) = \sum_{m=1}^M I_{\xm}(\theta) = M I_{\xo}(\theta). 
\end{equation*} 
It follows immediately that $I_{\xo}(\theta)=\frac{1}{M} I_X(\theta)$. Similarly, $I_{\xm}(\theta)=\frac{1}{M} I_X(\theta)$ for $m=2,\ldots,M$.
\end{proof}

\begin{lemma}
\label{lemma2}
In the setting of Lemma~\ref{lemma1}, the Fisher information contained in $\sum_{m=1}^K X^{(m)}$ for $K < M$, denoted $I_{\sum_{m=1}^K X^{(m)}}(\theta)$, is equal to $\frac{K}{M}I_X(\theta)$. 
\end{lemma}

\begin{proof}
The convolution-closed property of $F_\lambda$ says that $\sum_{m=1}^K X^{(m)} \sim F_{\frac{K}{M} \lambda}$. As $\sum_{m=1}^K X^{(m)}$ is a function of $X^{(1)}, \ldots, X^{(K)}$, we have that
\begin{equation*}
\label{main1}	
I_{\left(X^{(1)}, \ldots, X^{(K)}\right)}(\theta) \geq I_{\sum_{m=1}^K X^{(m)}}(\theta).
\end{equation*}
We can construct a random variable $Y \sim F_{\frac{K}{M} \lambda}$ with the same distribution as $\sum_{m=1}^K X^{(m)}$ by applying Algorithm~\ref{alg:datathin_manyfolds} to $X$ with $\epsilon_1 = K/M$ and $\epsilon_2 = 1-K/M$ and calling the first fold of data $Y$. We then can thin $Y$ with $\epsilon_1 = \ldots = \epsilon_K = \frac{1}{K}$ to obtain $Y^{(1)}, \ldots, Y^{(K)}$, whose joint and marginal distributions are equal to $X^{(1)}, \ldots, X^{(K)}$ (independent random variables that follow $F_{\frac{1}{M} \lambda}$). This allows us to rewrite the inequality above as  
\begin{equation}
\label{forward}
I_{\left(Y^{(1)}, \ldots, Y^{(K)}\right)}(\theta) = I_{\left(X^{(1)}, \ldots, X^{(K)}\right)}(\theta) \geq I_{\sum_{m=1}^K X^{(m)}}(\theta) = I_Y(\theta).
\end{equation}
By the same logic that was given in the proof of Lemma~\ref{lemma1}, we know that we can produce $I_{\left(Y^{(1)}, \ldots, Y^{(K)}\right)}$ from $Y$ without knowing $\theta$. Thus, $Y^{(1)}, \ldots, Y^{(K)}$ cannot contain more information about $\theta$ than $Y$. Thus, we have that 
\begin{equation}
\label{backward}
I_{\left(Y^{(1)}, \ldots, Y^{(K)}\right)}(\theta) = I_{\left(X^{(1)}, \ldots, X^{(K)}\right)}(\theta) \leq I_{\sum_{m=1}^K X^{(m)}}(\theta) = I_Y(\theta),
\end{equation}
where the last equality holds since two random variables with the same distribution contain the same amount of information about $\theta$. Combining the inequalities in \eqref{forward} and \eqref{backward} yields
$$
I_{\sum_{m=1}^K X^{(m)}}(\theta) = I_{\left(X^{(1)}, \ldots, X^{(K)}\right)}(\theta).
$$
Finally, we note that 
$$
I_{\sum_{m=1}^K X^{(m)}}(\theta) = I_{\left(X^{(1)}, \ldots, X^{(K)}\right)}(\theta) =  \sum_{m=1}^K I_{X^{(m)}}(\theta) = \frac{K}{M} I_{X}(\theta),
$$
where the second equality follows from independence, and the third from Lemma~\ref{lemma1}. This completes the proof. 
\end{proof}

We now proceed with the proof of Theorem~\ref{theorem_epsilon}. 
For simplicity, assume that $\epsilon_m$ is a rational number for $m=1,\ldots,M$. Then, we can rewrite $\epsilon_1,\ldots,\epsilon_M$ as $\frac{K_1}{M^*}, \frac{K_2}{M^*}, \ldots, \frac{K_M}{M^*}$ for integers $M^*, K_1, \ldots, K_M$. Further assume that $\frac{1}{M^*} \lambda$ is in the parameter space $\Lambda$ for our model. Then, for $m=1,\ldots,M$, $\xm$ has the same distribution as $\sum_{j=1}^{K_m} \tilde{X}^{(j)}$, where $\tilde{X}^{(1)}, \ldots, \tilde{X}^{(M^*)}$ are random variables that we would obtain if we thinned $X$ into $M^*$ equally sized folds. 
By Lemma~\ref{lemma2}, $I_{X^{(m)}}(\theta) = I_{\sum_{j=1}^{K_m} \tilde{X}^{(j)}}(\theta) = \frac{K_m}{M^*} I_X(\theta) = \epsilon_m I_X(\theta)$.

\begin{remark}
The proof of Theorem~\ref{theorem_epsilon} assumes that (i) $\epsilon_m$ is rational for $m=1,\ldots,M$, and that (ii) $\frac{1}{M^*} \lambda \in \Lambda$ for the common denominator $M^*$. Among the distributions in Table~\ref{tab:maintable}, only the $\mathrm{Binomial}(r,p)$ and the $\mathrm{Multinomial}(r,p)$ have relevant restrictions on the parameter space $\Lambda$. To be able to thin one of these distributions with parameter $\epsilon_m$, we must have that $\epsilon_m r$ is a positive integer. Thus, $\epsilon_m = \frac{Z_m}{r}$ for some integer $Z_m$, i.e. (i) is satisfied. Adopting the notation of the proof of Theorem~\ref{theorem_epsilon},  we note that $K_m=Z_m$ and that $M^*=r=\lambda$. Thus, the requirement that $\frac{\lambda}{M^*}$ is in the parameter space (i.e. that it is a positive integer) follows immediately, since $\frac{\lambda}{M^*}=1$. 
\end{remark}

\section{Simulation Study Supporting Details}
\label{sec:simpar}

In this section, we provide additional details about the simulation studies described in Section \ref{subsec:setup}.

For Example~\ref{subsubsec:binsetup}, in which we select the number of principal components for binomial data, we use the following setup. For $K^*=10$, we compute $\theta = UDV^T$ where $U$ is a $n \times K^*$ random orthogonal matrix, $D$ is a $ K^*\times K^*$ diagonal matrix with diagonal elements equal to $5,6,\dots,14$, and $V$ is a $d \times K^*$ random orthogonal matrix. Then, $p_{ij} = \frac{\exp{\left(\theta_{ij}\right)}}{1+\exp{\left(\theta_{ij}\right)}}$ for $i=1,\dots,n$ and $j=1,\dots,d$.

For Example~\ref{subsubsec:gammasetup}, in which we select the number of clusters in gamma-distributed data, we use the following setup.

In the small $d$, small $K^*$ clustering setting described in Example~\ref{subsubsec:gammasetup}, observations from each cluster are generated as $X_{ij} \overset{\mathrm{ind}}{\sim} \text{Gamma}(\lambda, \theta_{c_i,j})$ where $\lambda=20$,
$$
\theta = \begin{bmatrix} 0.5 & 5 \\ 5 & 0.5 \\ 10 & 10 \\ 0.5 & 0.5 \end{bmatrix},
$$
and $c_i \in \{1,2,3,4\}$ is the true cluster membership for the $i$th observation.

In the large $d$, large $K^*$ clustering setting described in Example~\ref{subsubsec:gammasetup}, observations from each cluster are generated as  $X_{ij} \overset{\mathrm{ind}}{\sim} \text{Gamma}(\lambda, \theta_{c_i,j})$ where $\lambda=2$, the $K^* \times d$ matrix $\theta$ is constructed such that for $j=1,\dots,d$ and $k=1,\dots,K^*$,
$$
\theta_{kj} = 
\begin{cases}
0.1 & \text{if } k \le 9 \text{ and } 10k-9 \le j \le 10k+10, \\
1 & \text{otherwise},
\end{cases}
$$
and $c_i \in \{1,2,\dots,10\}$ is the true cluster membership for the $i$th observation.

\section{Simulation with mean squared error loss function}
\label{sec:simMSE}

\subsection{Methods}

As an alternative to the negative log-likelihood loss used in Section \ref{sec:sim}, here we consider applying a mean squared error loss. 
To do this, we simply replace the negative log likelihood from Step 4 of Algorithms~\ref{alg:binpcaalg} and ~\ref{alg:gammaclustalg} with the mean squared error, defined as
\begin{equation}	
	\frac{1}{nd}\sum_{i=1}^n \sum_{j=1}^d \left(\xte_{ij} - \epsilon^{\mathrm{(test)}}rp_{ij}^{(K)}\right)^2
\end{equation}
in the case of Algorithm~\ref{alg:binpcaalg} and 
\begin{equation}	
\frac{1}{nd}\sum_{i=1}^n \sum_{j=1}^d \left(\xte_{ij} - \frac{\eps^{\mathrm{(test)}}}{\eps^{\mathrm{(train)}}}\hat\mu_{\hat c_i,j}^{\mathrm{(K)}}\right)^2
\end{equation}
in the case of Algorithm~\ref{alg:gammaclustalg}. 

After replacing the loss functions, these algorithms can be applied directly to obtain simulation results for data thinning, and with slight modification to obtain results for multi-fold data thinning and the naive method, as described in Section 4.2.  

\subsection{Results}

In Figure \ref{fig:role_MSE}, we plot the average mean squared error curves, as a function of $K$. As with the negative log-likelihood loss, data thinning approaches produce curves with sharp minimum values at or near $K=K^*$, as opposed to the naive method's monotonically-decreasing curves.

\begin{figure}
\includegraphics[width=\textwidth]{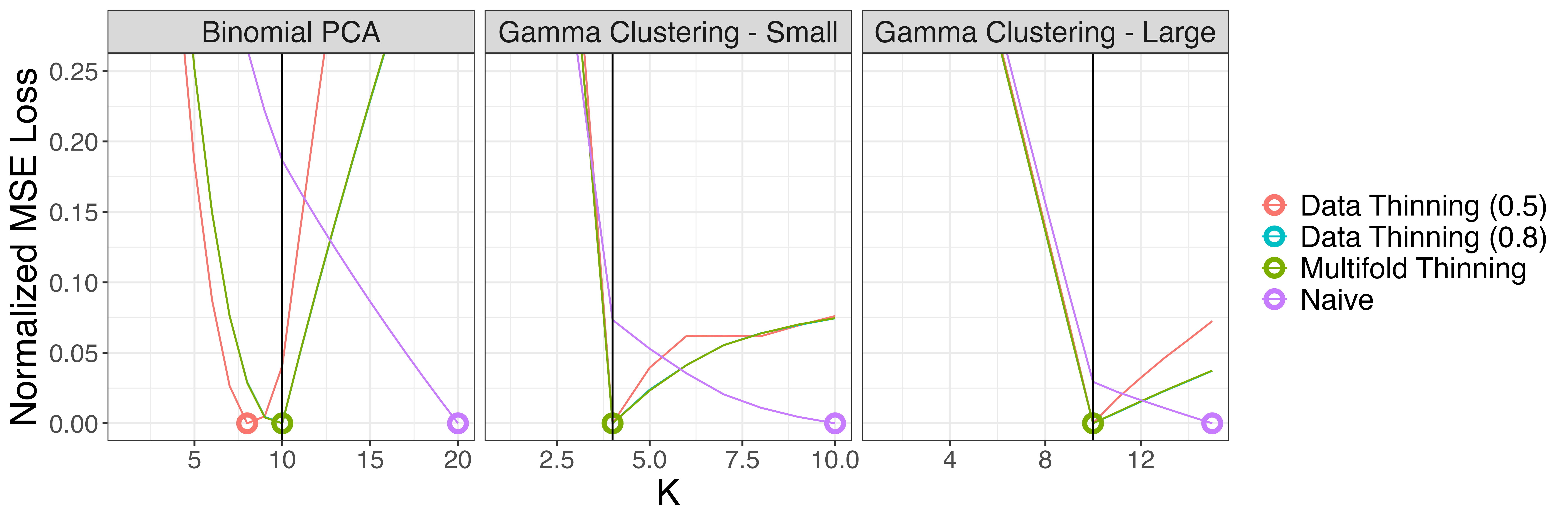}
\caption{The mean squared error loss averaged over 2,000 simulated data sets, as a function of $K$, for the naive method (purple), data thinning with $\epsilon^\mathrm{(train)}=0.5$ (red), data thinning with $\epsilon^\mathrm{(train)}=0.8$ (blue), and multifold thinning with $5$ folds (green). Each curve has been rescaled to take on values between $0$ and $1$, for ease of comparison. The minimum loss values for each method are circled, and $K^*$ is indicated by the vertical black line.}	
\label{fig:role_MSE}
\end{figure}

In Figure \ref{fig:role_eps2}, we plot the proportion of simulations that select the correct value of $K^*$ using the mean squared error loss, as a function of $\epsilon^\mathrm{(train)}$. Results are largely similar to Figure \ref{fig:role_eps}.

\begin{figure}
\includegraphics[width=\textwidth]{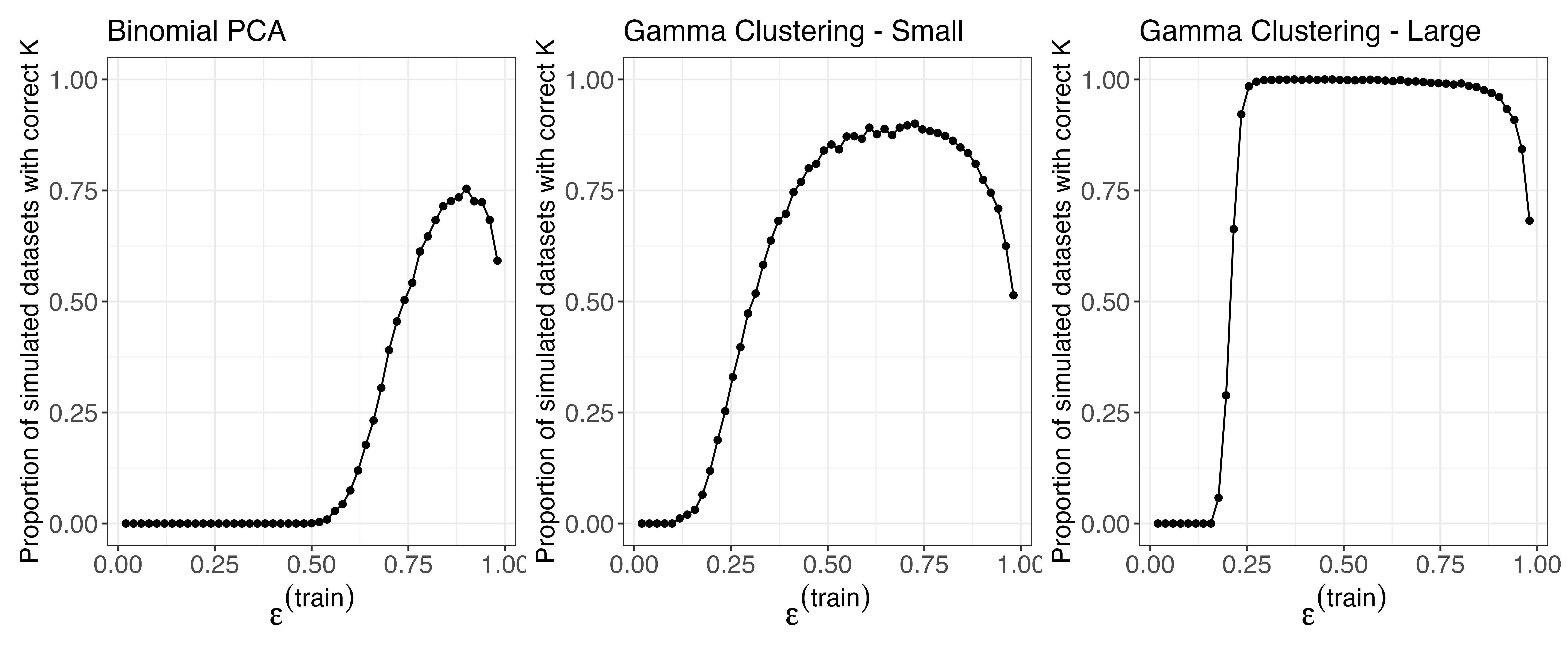}
\caption{The proportion of simulations for which data thinning selects the true value of $K^*$ with the mean squared error loss, as a function of $\epsilon^\mathrm{(train)}$, for the simulation study described in Section \ref{subsec:setup}. The optimal value of $\epsilon^\mathrm{(train)}$ depends on the problem at hand.}	
\label{fig:role_eps2}
\end{figure}

Finally, we compare multi-fold to single-fold thinning, under the mean squared error loss, in Figure \ref{fig:role_folds2}. As in Figure \ref{fig:role_folds}, we find that multi-fold thinning tends to select the correct value of $K$ more often than single-fold thinning.

\begin{figure}
\includegraphics[width=\textwidth]{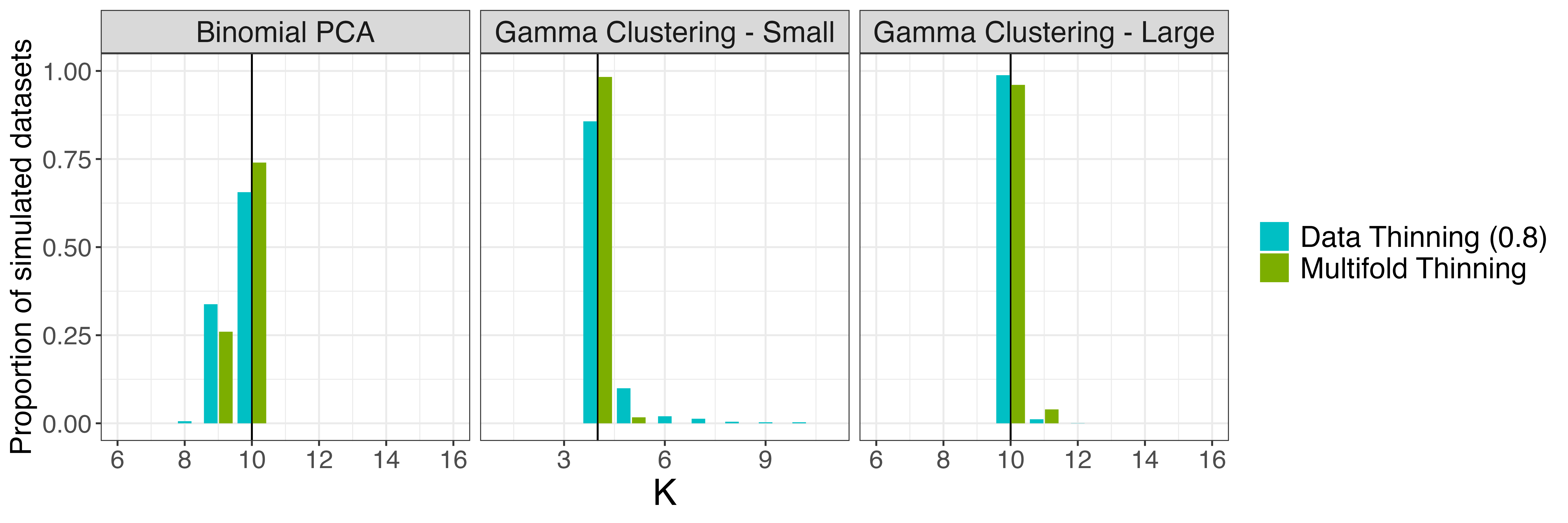}
\caption{The proportion of simulated data sets in which each candidate value of $K$ is selected, with the mean squared error loss, under data thinning with $\epsilon^\mathrm{(train)}=0.8$ (blue) and multifold thinning with $M=5$ (green), for each of the simulation settings described in Section \ref{subsec:setup}. The true value of $K^*$ is indicated by the vertical black line. Multifold thinning tends to select the true value of $K$ more often than single-fold thinning.}	
\label{fig:role_folds2}
\end{figure}

\section{Details for the real data analysis in Section~\ref{sec:data}}
\label{appendix:seurat}

We first explain in detail the preprocessing done to the matrix $X$ in the Seurat tutorial. 
\begin{list}{}{}
\item[(1)] Initial data filtering: We initially filter the data such that only cells with between 200 and 2500 total counts remain (with fewer than 5\% of the counts coming from from mitochondrial genes) and only genes that are expressed in at least 200 cells remain. This reduces the size of $X$ from  $2,700  \times 32,738$ to $2,638 \times 13,714$
\item[(2)] Log normalization: Next, the data are normalized and log transformed, such that
$$
Y_{ij} = \log \left( \frac{X_{ij}}{\sum_{t=1}^{13714} X_{it}} \times 10,000 +1 \right). 
$$
\item[(3)] {Feature selection:} Following this transformation, the top 2000 highly variable genes are selected using the function \texttt{FindVariableFeatures} from the Seurat package. The goal of the function is to find a subset of features with high cell-to-cell variation after accounting for the inherent mean-variance relationship, as these are most likely to be interesting in downstream analysis, and it implements methodology from \cite{seurat2}. 
\item[(4)] {Centering and scaling:} Finally, the columns of the subsetted matrix $Y \in \mathbb{R}^{2638 \times 2000}$ are centered and scaled to obtain the matrix $\tilde{Y}$. 
\end{list}
After these preprocessing steps, the principal components of $\tilde{Y}$ are computed.

We now explain the preprocessing for $\xo$ and $\xt$ that we use for our data thinning alternative to the Seurat tutorial. We follow the same four steps as above, but we are careful to specify what we do on the training set $\xo$ as opposed to the test set $\xt$. 
\begin{list}{}{}
	\item[(1)] Initial data filtering: We perform the initial data filtering from Step (1) above on $X^{(1)}$. We then subset $X^{(2)}$ to include the same genes and cells as those in $X^{(1)}$. After this step, $X^{(1)}$ and $X^{(2)}$ are both in $\mathbb{Z}_{\geq 0}^{2638 \times 13258}$. 
	\item[(2)]  Log normalization:  We normalize and log-transform both $\xo$ and $\xt$, such that:
$$
Y_{ij}^{(1)} = \log \left( \frac{X^{(1)}_{ij}}{\sum_{t=1}^{13258} X^{(1)}_{it}} \times 10,000 +1 \right)
, Y_{ij}^{(2)} = \log \left( \frac{X^{(2)}_{ij}}{\sum_{t=1}^{13258} X^{(2)}_{it}} \times 10,000 +1 \right).
$$
We note that these random variables are still independent and identically distributed under our Poisson assumption. 
\item[(3)] Feature selection:   We then apply the Seurat function \texttt{FindVariableFeatures} to the matrix $Y^{(1)}$ to select the top $2000$ highly variable genes \citep{seurat2}. We subset both $Y^{(1)}$ and $Y^{(2)}$ to contain only these genes, such that $Y^{(1)}, Y^{(2)} \in \mathbb{R}^{2638 \times 2000}$.
\item[(4)] Centering and scaling: We center and scale the columns of the subsetted $Y^{(1)}$ to obtain $\tilde{Y}^{(1)}$. We also center and scale the columns of the subsetted $Y^{(2)}$ to obtain $\tilde{Y}^{(2)}$. 
\end{list}
After these preprocessing steps, the principal components of $\tilde{Y}^{(1)}$ are computed, and the loss function is computed using $\tilde{Y}^{(2)}$.

We now explain the identity that makes Figure~\ref{fig:Seuratplot}(a) and Figure~\ref{fig:Seuratplot}(b) mathematically equivalent. In Section~\ref{sec:data}, we defined
$$
SSE_K(\tilde{Y}) = \left\| \tilde{Y} - U_{1:K} D_{1:K} V_{1:K}^T\right\|_F^2.
$$
We see that:
\begin{align*}
 \left\| \tilde{Y} - U_{1:K} D_{1:K} V_{1:K}^T\right\|_F^2 &= 
 \|\tilde{Y}\|_2^2 - 2 \mathrm{trace}\left(\tilde{Y}^T U_{1:K} D_{1:K} V_{1:K}^T \right) + \mathrm{trace}\left( V_{1:K} D^T_{1:K} U_{1:K}^T U_{1:K} D_{1:K} V_{1:K}^T  \right)  \\
 &= \|\tilde{Y}\|_F^2 - \mathrm{trace}\left(  D_{1:K}^T D_{1:K} \right)  = \left\| \tilde{Y} \right\|_F^2 - \sum_{j=1}^K D_{jj}^2. 
\end{align*}
Thus, if we compute $SSE_K(\tilde{Y})$ for $K=1,\ldots,20$, since $\left\| \tilde{Y} \right\|_F^2$ is fixed, we can easily obtain the values of $\sum_{j=1}^K D_{jj}^2$ for $K=1, \ldots, 20$. By taking the differences between these values for $K$ and $K+1$, we obtain 
$D_{KK}$ for $K=1,\ldots, 20$. Finally, we note that the standard deviation of the $K$th principal component $(U_{K} D_{KK})$ can be written as
$$
\sqrt {( D_{KK} U_{K})^T (D_{KK} U_{K} )} = \sqrt{ D_{KK}^2 }.
$$

Since the standard deviation of the $K$th principal component (plotted in Figure~\ref{fig:Seuratplot}(a)) can be obtained directly from the sums of squared error plotted in Figure~\ref{fig:Seuratplot}(b), we say that the two plots are mathematically equivalent.

\vskip 0.2in
\bibliography{datathin.bib}

\end{document}